\documentclass[conference]{IEEEtran}
\IEEEoverridecommandlockouts
\usepackage{amsmath,amssymb,amsfonts,mathtools, amsthm}
\usepackage{algorithmic}
\usepackage{graphicx}
\usepackage{enumitem} 
\usepackage{textcomp}
\usepackage{capt-of} % here remove if problem with figures

\usepackage{xcolor}
\usepackage{bm}                        
\usepackage{dsfont}  
\usepackage{color}
\usepackage{physics}

\usepackage[caption=false,font=footnotesize]{subfig}

\newtheorem{thm}{Theorem}
\newtheorem{theorem}{Theorem}
\newtheorem{lemma}[thm]{Lemma}
\newtheorem{remark}[thm]{Remark}

\newcommand{\panos}[1]{{\textcolor{cyan}{Panos: #1}}}

\def\BibTeX{{\rm B\kern-.05em{\sc i\kern-.025em b}\kern-.08em
    T\kern-.1667em\lower.7ex\hbox{E}\kern-.125emX}}
\begin{document}

% 1. Define your "if" switch
\newif\ifjournal
% 2. Set your switch (choose one)
%\journaltrue   % <-- Use this to SHOW the journal notes
\journalfalse  % <-- Use this to HIDE the journal notes

% here extension flag
\newif\ifextension
%\extensiontrue   
\extensionfalse

% \ifjournal
% \else
% \fi

\title{Rate-Fidelity Tradeoffs in All-Photonic and Memory-Equipped Quantum Switches
}

\author{Panagiotis~Promponas, %~\IEEEmembership{Student~Member,~IEEE,}
Leonardo~Bacciottini, Paul~Polakos, Gayane~Vardoyan, \\ Don~Towsley,
and~Leandros~Tassiulas%,~\IEEEmembership{Senior~Member,~IEEE}
\vspace{-\baselineskip}
\thanks{This work was supported by the Army Research Office MURI under the project number W911NF2110325, by the QuantumCT project, and by the National Science Foundation under project number CNS 2402862.}
\thanks{P. Promponas (panagiotis.promponas@yale.edu) and L. Tassiulas are with the Department of Electrical and Computer Engineering, Yale University, New Haven, CT, USA.}
\thanks{L. Bacciottini, G.
Vardoyan and D. Towsley are with the University of Massachusetts, Amherst, MA, USA.}
\thanks{P. Polakos is with Einblick LLC, Marlboro, NJ, USA.}

}% <-this % stops a space

\maketitle

\begin{abstract}
Quantum entanglement switches are a key building block for early quantum networks, and a central design question is whether near-term devices should use only optical components and photonic qubits or also incorporate quantum memories. We compare two architectures: an all-photonic entanglement generation switch (EGS) that repeatedly attempts Bell-state measurements (BSMs) without storing qubits, and a quantum memory-equipped switch that buffers entanglement and triggers measurements only when heralded connectivity is available (\emph{herald-then-swap} control). These two designs trade off simple, memoryless operation that avoids decoherence and memory-induced latency against heralding-based control that buffers entanglement to use BSMs more efficiently.
% The EGS operates continuously without quantum memories, avoiding decoherence and latency but wastes Bell-state measurements (BSM) attempts due to the lack of conditioning on heralded connectivity. The quantum memory-equipped switch buffers qubits and schedules BSMs based on heralded connectivity, improving resource utilization at the cost of latency and memory noise. 
% We formalize both models under a common hardware abstraction and characterize their achievable rate--fidelity regions under realistic constraints. We then introduce a benchmarking methodology that links hardware parameters to network-level performance.
We formalize both models under a common hardware abstraction and characterize their achievable rate–fidelity regions, yielding a benchmarking methodology that translates hardware and protocol parameters into network-level performance. Numerical evaluation quantifies the rate–fidelity tradeoffs of both models, identifies operating regions in which each architecture dominates, and shows how hardware and protocol knobs can be tuned to meet application-specific targets.
% Numerical results show that, in the regimes considered, quantum memory-equipped switches are preferable mainly in a well-defined high-fidelity band (e.g., \(0.75 \lesssim F \lesssim 0.9\)), where they achieve higher rates at fixed target fidelity,  while all-photonic architectures deliver higher utility outside this band, yielding clear, scenario-dependent design guidance.
\end{abstract}

%Rate–Fidelity Tradeoffs in All-Photonic vs quantum memory-equipped Quantum Entanglement Switches

%CAMERA-READY
% \begin{IEEEkeywords}
% Quantum Switch; Entanglement Generation Switch; Scheduling; Throughput Region
% \end{IEEEkeywords}

\section{Introduction}

Quantum networks that distribute high-fidelity entanglement
between distant users underpin applications in secure communication,
distributed sensing, and distributed quantum computing~\cite{van2014quantum,caleffi2024distributed}.
A central challenge is that in lossy optical channels, attenuation causes the success probability of each entanglement-generation attempt to decrease exponentially with distance. This limitation introduces the need for intermediary network nodes, such as quantum repeaters \cite{briegel1998quantum} and switches \cite{vardoyan2021stochastic}, that employ entanglement swapping to establish entanglement between non-adjacent clients.

% The evolution from linear repeater chains to multi-user quantum
% switches in star or mesh topologies marks a critical transition:
% the main challenge shifts from physical-layer error management
% on a single path to network-layer resource allocation, scheduling,
% and contention among multiple user pairs. 

The evolution from quantum repeaters to multi-user quantum switches introduces a new layer of network-level complexity: the switch must allocate shared resources and schedule entanglement swapping under contention among many user pairs, on top of the underlying physical-layer error processes. The quantum switch is a network node connected to three or more other nodes. It functions as a two-sided queueing system \cite{bhambay2025optimal}: on one side, requests for end-to-end entanglement arrive, while on the other, probabilistic link-level entanglement generation (LLEG) produces Bell pairs \cite{einstein1935can} that must be consumed through entanglement swaps to serve requests.

% There is a vast number of papers in the literature, both in the theoretical and experimental communities, that consider a multitude of different quantum switching models. 

A large body of work considers
a variety of quantum switching models \cite{vardoyan2021stochastic, gauthier2023architecture}. A useful taxonomy based on the principal architectural tradeoff between hardware simplicity and control intelligence is rooted in the use of quantum memories and the resulting operational logic. This categorization provides a framework for grasping the landscape of current research and motivates the selection of the models analyzed in this paper.

The first major class of switch architectures comprises all-photonic hardware \cite{azuma2015all} that performs entanglement swapping directly on arriving photons, without storing qubits in quantum memories at the switch. In these architectures, often termed entanglement generation switches (EGS) \cite{gauthier2023architecture, gauthier2024demand, yau2025service}, the switch, in the absence of quantum memories, cannot herald successful link-level entanglement (LLE) attempts prior to performing a Bell-state measurement (BSM). Moreover, BSMs are carried out by a set of Bell-state analyzers (BSAs), each of which can be assigned to any user pair. The lack of heralding forces the switch to trigger BSMs “blindly” on preassigned pairs. Hence, a BSA is often triggered when at most one of the two clients has a successful LLE (i.e., an entangled photon reaches the corresponding BSA), so no end-to-end entangled pair is produced. Recent work explores resource allocation schemes for EGSes, including on-demand blocking models that drop requests when resources are unavailable \cite{gauthier2024demand},
and load-balancing policies that poll nodes for requests \cite{yau2025service}.

% The first major class of switch architectures comprises all-photonic hardware \cite{azuma2015all} that performs entanglement swapping directly on arriving photons, without storing qubits in switch quantum memories. In these architectures, often termed entanglement generation switches (EGS) \cite{gauthier2023architecture, gauthier2024demand, yau2025service}, the switch, in the absence of quantum memories, cannot herald successful LLE attempts prior to performing a Bell-state measurement (BSM). In the EGS model, BSMs are performed from a number of Bell State Analyzers (BSA). Each BSA can be assigned to any pair of user nodes. The lack of heralding forces the switch to trigger BSMs blindly on preassigned user pairs. As a result, a BSA is often triggered in times when at most one of the two connected clients has a successful LLE (i.e., a photon successfully reaching the switch), so no end-to-end pair is produced. Recent work explores resource allocation schemes for EGSes, including on-demand blocking models that drop requests when resources are unavailable \cite{gauthier2024demand},
% and load-balancing policies that poll nodes for requests \cite{yau2025service}.

The second class of architectures incorporates quantum memories to enable more sophisticated, dynamic control logic. The operational paradigm of these switches is ``herald-then-swap". In this model, LLEG is first attempted between each client and the switch. The success of each attempt is signaled to the switch's controller via a ``herald" message. The switch makes a scheduling decision after this connectivity-acquisition phase. The controller, now aware of the set of available Bell pairs, can perform an optimal matching, pairing occupied memories to entanglement swapping operations. This approach enables more efficient use of the stored Bell pairs, since BSM attempts are applied only to heralded links. The memory-enabled switch has been the subject of extensive analysis characterizing its throughput region, developing scheduling policies \cite{vasantam2022throughput, tillman2024calculating, dai2021entanglement, promponas2024maximizing}, and analyzing the impact of memory decoherence \cite{valls2023capacity, panigrahy2023capacity, bhambay2025optimal}.

%zubeldia2025matching

% Choosing between these two switching paradigms reflects a fundamental tradeoff between information-driven efficiency and latency-induced decoherence. 
% As a result, neither model is strictly superior. This architectural tension motivates our comparative framework. To the best of our knowledge, this is the first work to directly quantify and compare the rate–fidelity tradeoffs of each model under realistic hardware assumptions.

These two architectures raise a concrete design question: for a given hardware stack and target application, should one deploy an all-photonic or a memory-equipped switch? Quantum memories can, in principle, boost end-to-end rates and fidelities, but if coherence times are short or heralding delays are large, a memory-based design can lose this advantage and even underperform compared to a simpler all-photonic switch. Thus, architectural superiority is regime-dependent rather than absolute. This paper makes that choice explicit by placing both models in a common abstraction and quantifying their achievable rate–fidelity tradeoffs under realistic hardware assumptions. Our key contributions are:

% \subsection{Contributions}

% This paper develops a rigorous comparative framework for near-term quantum switches by formalizing, analyzing, and benchmarking two distinct architectural abstractions: an all-photonic entanglement generation switch (EGS) and a quantum memory-equipped switch based on herald-then-swap control. 

\begin{itemize}[leftmargin=*]
\item We define mathematically precise models for both all-photonic and quantum memory-equipped switches, specifying their control strategies, physical resource assumptions, and operational logic. 

%     \item We derive closed-form characterizations of the achievable throughput rates and fidelity under each model, explicitly quantifying how hardware parameters impact end-to-end entanglement performance.

% \item We introduce a benchmarking methodology, that provides a unified basis for comparing these architecturally distinct models. This framework directly links low-level hardware parameters to network-level performance, enabling informed hardware/protocol co-design and identifying the specific operational regimes where one architecture is superior.
    
\item We derive closed-form characterizations of the achievable throughput and fidelity under each model, explicitly quantifying how hardware parameters impact end-to-end entanglement performance. Building on this, we develop a benchmarking methodology that provides a unified basis for comparing the two architectures, links low-level hardware parameters to network-level performance, and identifies operational regimes in which each architecture is superior.

\item Numerical evaluation quantifies the rate-fidelity trade-offs of both models, and shows how hardware and protocol knobs can be tuned to meet application-specific targets.
    % in the regimes considered, quantum memory-equipped switches are preferable mainly in a well-defined high-fidelity band (e.g., \(0.75 \lesssim F \lesssim 0.9\)), while all-photonic architectures deliver higher utility outside this band, providing scenario-dependent guidance.
    % % in the examined regimes, the quantum memory-equipped model is advantageous only for stricter fidelity requirements (e.g., $F \gtrsim 0.75$), whereas the EGS model is more favorable for lower fidelities.
\end{itemize}

\ifjournal
+ future work
Introduce some new models (in terms of end-to-end entanglement rate analysis) and their benefits/drawbacks (e.g., delay lines instead of memories, reallocatable memories, sources both in nodes and switches). Introduce the possibility of hybrid models. This contribution bullet would probably be on an extension and not in the conference deliverable.
Check also for future work \cite{azuma2015all} that uses a more convoluted protocol for EGS all photonics.
\else
\fi

\section{Preliminaries}
\label{sec:preliminaries}

\subsection{System Description}\label{subsec:system_description}

We consider a star-topology quantum network with a central switch node connected to $N$ client nodes. Let $\mathcal N$ be the node set and $\mathcal F\triangleq\{(i,j): i<j,\ i,j\in\mathcal N\}$. We will use $\mathcal F$ to avoid double-counting node pairs and refer to each $(i,j)\in\mathcal F$ as an end-to-end flow. The total number of flows is thus $F \triangleq \binom{N}{2} = |\mathcal F|$. The switch generates and distributes entangled qubit pairs (Bell pairs) between the clients. Each node is equipped with quantum memories to store qubits until end-to-end entanglement is established. The switch can perform at most $B$ BSMs in parallel to generate end-to-end Bell pairs due to hardware limitations (e.g., it has $B$ BSAs). 
We assume time-slotted operation. The detailed timing of LLEG attempts, BSMs, and heralding is model-dependent and is specified separately for the EGS and memory-equipped switch models.

\subsection{Fidelity and Quantum Noise Model}\label{subsec:werner_states}
% Entanglement fidelity is a measure of how close a generated entangled state is to the ideal target entangled state. 
% Diverse applications may impose in general different fidelity requirements.
% To model noisy entangled states, we use \emph{Werner states}, which are a specific type of mixed state that represent the effects of noise and imperfections using one parameter $w$. A Werner state for a pair of qubits can be expressed as:

Entanglement fidelity quantifies how close a generated state is
to a target Bell state, and different applications impose different
fidelity requirements. We model noisy entangled states as Werner
states, a one-parameter family of mixed states that captures
imperfections via the Werner parameter $w \in [0,1]$. A Werner state for a pair of qubits is expressed as:  \vspace{-5pt}
\begin{equation*}
\rho_w = w \ket{\Phi^+}\bra{\Phi^+} + (1-w) \frac{\mathbb{I}}{4},    \vspace{-5pt}
\end{equation*}
where $\ket{\Phi^+} = \frac{1}{\sqrt{2}}(\ket{00} + \ket{11})$ is one of the reference Bell states, $\mathbb{I}/4$ is the completely mixed state, and $w \in [0,1]$ is the Werner parameter. The fidelity of such a state is $F_w = \frac{3w+1}{4}$. When two Werner states with Werner parameters $w_1$ and $w_2$ are swapped with perfect operations, the resulting state's Werner parameter is simply the product $w_1 \cdot w_2$. Even if the actual entangled state is not a Werner state, such an approximation provides lower bounds for fidelity analysis under an appropriate set of assumptions \cite{vardoyan2023quantum}.

A depolarizing channel with parameter $q \in [0,1]$ acting on a Werner state simply replaces its Werner parameter $w$ with $w \cdot q$. We model the fidelity degradation of the swapped pair using a depolarizing channel with parameter $q_{\mathrm{BSM}}$. Independently, the probability that a BSM attempt succeeds is captured by $p_{\mathrm{BSA}}$ (for BSAs in the EGS model) or $p_{\mathrm{swap}}$ (for the memory-assisted model); these parameters affect the end-to-end generation rates rather than the output-state fidelity. We also model quantum decoherence in the memories as a depolarizing channel with parameter $e^{-\tau/T}$, where $T$ is the quantum memory coherence time (assumed equal for all memories), and $\tau$ is the time spent in memory. 

% Thus, consider a repeater-style step where two input Werner pairs with parameters $w_1$ and $w_2$ are created, with one qubit of each stored in memories at the end nodes for times $\tau_1$ and $\tau_2$ while their partners travel through fiber to a central BSA, then the Werner parameter of the resulting end-to-end Bell pair is
% \vspace{-5pt}
% \begin{equation} \vspace{-5pt}
% w_{\text{e2e}} = w_1 \cdot w_2 \cdot q_{\mathrm{BSM}} \cdot e^{-(\tau_1 + \tau_2)/T}.\label{eq:e2e_werner_parameter}
% \end{equation}
\subsection{Link-level Entanglement Generation Model}
\label{ssec:prelim_LLE_model}
%The physical system used to realize a qubit can, depending on its nature, act as a \emph{storage qubit} (i.e. a quantum memory), or as a \emph{flying qubit} (i.e. a traveling carrier). Specifically, photons are the natural choice for flying qubits, while matter-based systems (e.g., trapped ions, quantum dots, color centers, atomic ensembles, etc.) are typically used as storage qubits. Typically, a BSM operation between two flying qubits is performed using linear optics \cite{munro2015inside}, which is \emph{probabilistic in nature} with a success probability usually close to $0.5$. On the other hand, BSMs between storage qubits can be deterministic, but require the presence of quantum memories and control logic at the measuring node.

% We denote by \emph{LLE} the protocol that distributes and heralds Bell pairs between a client and the switch. Its characteristics depend on the underlying hardware platform; nonetheless, most implementations share a common structure. All the LLE protocols rely on sources (e.g., Symmetric Parametric Down-Conversion (SPDC) \cite{alshowkan2021reconfigurable}, or emitters \cite{barrett2005efficient} based on quantum dots, color centers, or trapped ions) that emit entangled photons traveling on the physical link between the client and the switch. LLE protocols are inherently probabilistic and require some form of heralding. These similarities enable a hardware-independent abstraction described by a common set of parameters. 

We denote by link-level entanglement generation (LLEG) the protocol that generates and distributes Bell pairs between a client and the switch; the resulting Bell pair is referred to as LLE. While details depend
on the hardware platform, most implementations share a common
structure: they use sources (e.g., Spontaneous Parametric Down-Conversion (SPDC)~\cite{alshowkan2021reconfigurable} or emitters~\cite{barrett2005efficient}
based on quantum dots, color centers, or trapped ions)
that send entangled photons along the client--switch link. LLEG
is inherently probabilistic and requires heralding. To differentiate between this heralding and the one needed to establish an end-to-end entanglement between distant nodes, we call this link-level heralding. A hardware-independent abstraction can be described by a common set of parameters. For LLEG between $i$ and the switch, define:

\textbf{Success probability $p_i$}: Probability that an LLEG attempt successfully produces a Bell pair between $i$ and the switch. 
% It depends on the  link length, detector efficiencies, etc.

\textbf{Heralding overhead $\tau_{\mathrm{hrld}}$}: The \emph{overhead} time needed \emph{after} a successful LLE has been established (e.g., one memory in a node and the switch share an entangled pair) to herald the link-level entanglement in the switch (so it can learn connectivity). Node-side confirmation for the LLE can occur in parallel and does not affect switch operation, since nodes ultimately learn success/failure via the end-to-end entanglement heralding.
% \textbf{Heralding overhead $\tau_{\mathrm{hrld}}$}: The time from when the first half of a newly created LLE is available at one end, until the LLE is fully established across the link and the switch has received the classical notification confirming this link-level entanglement. Node-side confirmation for the LLE can occur in parallel and does not affect switch operation, since nodes ultimately learn success/failure via the end-to-end entanglement heralding.

\textbf{Werner parameter $w_{i}$}: The Werner parameter associated with a successful LLE of node $i$. Hence, $w_{i}$ also accounts for decoherence on the qubits involved in this LLE while the switch awaits for the link-level heralding signals.
% The Werner parameter of a Bell pair generated between client $i$ and the switch after a successful LLE generation. 
% In a more specific model, this parameter could be a full density matrix for the resulting state.

% In particular, $w_{i}(\beta)$ absorbs all link-local noise and storage-induced decoherence, including any waiting time for classical heralding implied by the chosen LLE protocol and source placement.

\textbf{Tuning parameter $\beta$}: A controllable property of the source (e.g., pump power for SPDC or brightness for emitters \cite{vardoyan2023quantum}) that affects $p_i$ and $w_{i}$. Increasing $\beta$ raises $p_i$ but lowers $w_{i}$ due to multi-pair emissions. Thus, $p_i$ and $w_{i}$ are functions of $\beta$, denoted as $p_i(\beta)$ and $w_{i}(\beta)$.

% \textbf{Repetition time $\Delta t$}: Time between consecutive LLE attempts (e.g. pulse rate, $f_{\mathrm{pulse}} \equiv 1 / \Delta t$ for SPDC). 

\textbf{Attempt period $\Delta t$:}
Minimum time separation between the start of two consecutive LLEG attempts on a given link. For example, with an SPDC source that makes a single
pair-generation attempt per optical pulse, $\Delta t = 1/f_{\text{pulse}}$, where $f_{\text{pulse}}$ is the source pulse repetition rate (Hz).

\textbf{Multiplexing degree $S_i$}: Number of parallel LLEG attempts that can be performed between client $i$ and the switch. 
% The LLE can thus distribute at most $S_i$ Bell pairs per time slot. 
It captures spatial or frequency multiplexing, e.g., using several sources or a frequency-multiplexed SPDC source \cite{chen2023zero}. We refer to each such parallel attempt as a \emph{multiplexing frame}.

Well-known LLEG protocols that fit this abstraction\footnote{The discrete attempt abstraction is exact for pulsed sources \cite{kok2000postselected}. For continuous wavelength sources \cite{alshowkan2021reconfigurable}, where pair inter-generation times are approximately Poisson with rate $\lambda(\beta)$, $p_{i}(\beta)=1-e^{-\lambda(\beta)\Delta t}\approx \lambda(\beta)\Delta t$ for $\lambda(\beta)\Delta t\ll1$. We treat any success within a window as a Bernoulli trial registered at the slot boundary to fit in a discrete setting. } include, for example, the three schemes based on SPDC sources described in \cite{jones2016design}, the Barret-Kok protocol \cite{barrett2005efficient} using emitters, and the single click scheme as described in \cite{vardoyan2023quantum}. Some of these protocols are further discussed in later sections, as we can further specialize our switch models based on the specific LLEG protocol used. Finally, we note that not all these LLEG protocols are suitable for all switch models described in the following sections. For example, some protocols (e.g., the midpoint source protocol with block-based attempts from \cite{jones2016design}) may require memories at the switch, which is not compatible with the all-photonic switch model in Section \ref{sec:model_EGS}.

\section{Logical Model~1: All-Photonic Switch 
%(Port-Binding Control)
}
\label{sec:model_EGS}

%\subsection{Motivation and role}

This model of an \emph{entanglement generation switch (EGS)} captures an all-photonic switch with no quantum memories, $N$ client nodes, multiple entanglement sources and an in-switch pool of $B$ \emph{BSAs}, that can be \emph{re-bound} to port pairs. 
% only periodically at calibration epochs. 
The BSMs are probabilistic and their outcomes are heralded back to the endpoints, but the switch does not perform any link-level heralding before attempting the BSMs. The scheduler’s core decision is {\emph{which port pair each BSA listens to}} so that incoming photons are routed to the same BSA. Since client nodes may have multiple sources (thus, a multiplexing degree $S_i$), several BSAs can simultaneously serve the same clients. 
Similar switching models appear in \cite{gauthier2023architecture, gauthier2023control, gauthier2024demand, yau2025service}.

\subsection{System Formulation}
\label{ssec:EGS_system_formulation}
The LLEG protocol between the switch and client \(i\) relies on $S_i$ independent optical paths (e.g., fiber optic modes) into the switch’s ingress ports. Time is slotted with duration \(\Delta t\). Hardware calibration enforces a grouping of slots into epochs of \(T_{\mathrm{freeze}}\) slots each.
% (duration \(T_{\mathrm{freeze}}\Delta t\))
Each such interval is called a \emph{freeze epoch}, during which the control decision is fixed and updated only between epochs \cite{gauthier2024demand}. We assume a constant $T_{\mathrm{freeze}}$ and negligible calibration time at epoch boundaries; if this overhead is non-negligible, all per-second rates derived later should be scaled by the corresponding active-duty factor (i.e., by the fraction of time spent in freeze epochs).

% Time is slotted with slot duration \(\Delta t\); hardware calibration imposes \emph{freeze epochs} \don{What is a freeze window? Did you define it?}of \(T_{\mathrm{freeze}}\) slots (duration \(T_{\mathrm{freeze}}\Delta t\)). Control decisions are updated only between epochs and remain constant within an epoch.

We set the slot duration to the LLEG attempt period and
make two distinct modeling assumptions about heralding. First, the switch does not wait for \emph{link-level} heralding before attempting BSMs: BSAs are committed to node pairs within an epoch and attempt BSMs blindly on the incoming photons. Second, at the \emph{end-to-end} level, we assume that user nodes have sufficiently many local memories to buffer all
tentative Bell-pair halves across multiple slots until the classical
BSM outcomes are heralded, focusing only on the resource constraints of the switch. Under this  assumption, the start of the next slot is not delayed by end-to-end heralding, and the slot duration is not constrained by classical communication times.\footnote{If the BSA is slower than the source, or scarce node memories require waiting for end-to-end heralding before reuse, one can simply take the slot duration to be the dominant of these times. This only rescales the effective switching frequency and thus all per-second rates by a constant factor. Comparing the efficacy of the EGS architecture with and without the abundance-of-memories assumption for the nodes is left for future work.}

% We set the slot duration to the LLE repetition time and \egs{do not wait for link-level heralding before BSMs: since
% BSAs are committed to node pairs, they
% attempt BSMs blindly on the incoming photons without
% knowing which LLEs succeeded, and the resulting BSM
% outcomes are later heralded to the endpoints.} 
% We assume end nodes have sufficient memories to buffer their qubits until heralding occurs and focus on the constraints of the switch.
% \footnote{If the BSA is slower than the source, redefine the slot accordingly; the analysis is unchanged.}

\textbf{Control decisions:} We assume that the BSAs in the switch, the per-node sources, and the physical links are homogeneous. Therefore,
the scheduler’s decision within a freeze epoch $e$ is fully captured by the \emph{station-allocation} variables $x_{ij}\in\mathbb{Z}_{\ge 0}\quad (i,j) \in \mathcal F,$
denoting how many BSAs are reserved to serve client pair $(i,j) \in \mathcal F$ during epoch $e$. Feasibility is:
% \vspace{-4pt}
% \begin{equation}
% \sum_{i<j} x_{ij} \le B, 
% %\label{eq:station-budget} 
% \qquad \sum_{j\neq i} x_{ij} \le S_i \qquad (\forall i) %\label{eq:source-budget}
% \label{eq:X_set_definition_budget_source_constraints}
% \end{equation} 
\vspace{-4pt}
\begin{equation}
\sum_{(i,j)\in \mathcal{F}} x_{ij} \le B, 
%\label{eq:station-budget} 
\qquad \sum_{j<i} x_{ji} + \sum_{j>i} x_{ij} \le S_i \quad (\forall i). %\label{eq:source-budget}
\label{eq:X_set_definition_budget_source_constraints}
\end{equation} 
i.e., a global BSA budget and per-node source budgets.  Let us denote the convex set defined by inequalities \eqref{eq:X_set_definition_budget_source_constraints} as $\mathcal{X}$.

\textbf{Per-link LLEG success:} Per Sec.~\ref{ssec:prelim_LLE_model}, multiple LLEG protocols are possible. To illustrate the computation of $p_i(\beta)$, we consider a midpoint-source
architecture (between node and switch) with a single photon-pair attempt per slot.
% Hereafter, we assume midpoint sources (between node and switch). The framework remains agnostic: other source placements or protocols fit by re-deriving the parameters of Sec.~\ref{ssec:prelim_LLE_model}. 
Let $p_{\mathrm{pair}}(\beta)$ denote the probability that the source successfully produces two entangled photons. We expect that, $p_{\mathrm{pair}}(\cdot)$ is a strictly increasing function with respect to $\beta$ (see Sec.~\ref{ssec:prelim_LLE_model}).

An LLEG attempt succeeds iff the photon latches onto the node memory and the partner photon reaches the BSA $b$ within the BSM acceptance. The probability of a successful LLEG attempt is thus $p_i(\beta)
= p_{\mathrm{pair}}(\beta)\,\big(\eta_{\mathrm{m}} g_{\mathrm{m}}\big)\,\big(\eta_{\mathrm{sw}} g_{\mathrm{sw}}\big).$ Here $\eta_{\mathrm{m}}$ captures half-link and coupling losses to the node memory, $g_{\mathrm{m}}$ the memory write/gate efficiency; $\eta_{\mathrm{sw}}$ captures half-link and routing to any station $b$, $g_{\mathrm{sw}}$ the temporal/spectral/mode acceptance at the BSM input.

\textbf{End-to-end entanglement generation success:}
With exactly one LLE on each side of $b$, a coincidence attempt occurs in a slot iff \emph{both} sources succeed and the photons are captured. 
Thus, for a BSA $b$ serving $(i,j)$ in epoch $e$, \vspace{-5pt}
\begin{equation}  
p^{\mathrm{e2e}}_{ij}(\beta)\;=\;\bigl(p_i(\beta)\,p_j(\beta)\bigr)\;\xi^{2}\;p_{\mathrm{BSA}}.\label{eq:e2e_prob_egs} \vspace{-5pt}
\end{equation}
Here $\xi$ is the detector efficiency of every station $b$ and $p_{\mathrm{BSA}}$ is the success probability of the optical Bell state analyzers. Below we sum up the operation of this type of switch:

\textbf{Time-slot operation:} \emph{(i) Epoch allocation:} At the start of an epoch, choose $x_{ij}(e)$ satisfying inequalities \eqref{eq:X_set_definition_budget_source_constraints} and configure any $x_{ij}(e)$ stations to serve $(i,j)$ for the epoch.
\emph{(ii) Attempt:} Every node $i$ performs up to $S_i$ LLE attempts with identical success probability $p_i(\beta)$.
% \emph{(iii) Coincidence:} A coincidence occurs with probability $p_{ij}^{\mathrm{coin}}(\beta)
% = p_i(\beta)p_j(\beta)$ for each configured station serving $(i,j)$.
\emph{(iii) BSM:} 
BSMs are implemented with success probability (under homogeneity) $p^{\mathrm{e2e}}(\beta) \triangleq p^{\mathrm{e2e}}_{ij}(\beta)$ for every pair $(i,j)$.
\emph{(iv) Herald:} Upon success, the BSA sends the classical outcome to $i,j$. We assume that the next time slot begins immediately after the BSMs (step (iii)). That means that heralding occurs simultaneously and meanwhile each client node stores a Bell pair half in a dedicated quantum memory. For a graphical illustration of the EGS as well as its time-slot operation see Fig.~\ref{fig:main_illustration_switches}(a) and~\ref{fig:main_illustration_switches}(c) respectively.

\begin{figure*}[t]
  \centering
\includegraphics[width=.85\textwidth]{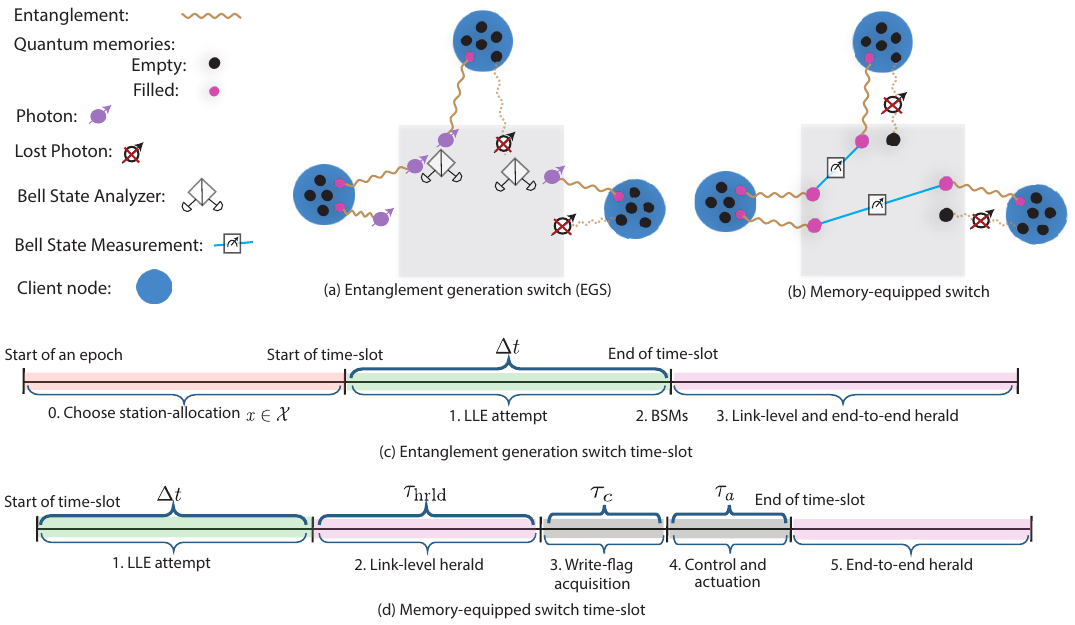}
  \caption{Graphical illustration of (a) an EGS, (b) a memory-equipped switch, (c) the time-slot operation of an EGS, and, (d) the time slot operation of a memory-equipped switch. Both of the switches in illustrations (a) and (b) use $S_i = 2$ for every $i$, and $B = 2$. For the memory-equipped model in (b), $M_i = 2$ for every node $i$.}
  \label{fig:main_illustration_switches}
\end{figure*}

% \leo{Let's go over these together in next meeting.} \panos{We can even delete it if we do not need it.}

% \textbf{Feasible physical architectures (examples that fit under this abstract switch model).}
% \begin{itemize}
% \item \emph{Fixed sources per node:} SPDC (PPLN/PPLT) or SFWM (SiN) arrays with identical pump clocks; QD emitters with electronic gating; spectral/temporal filtering for indistinguishability.
% \item \emph{Programmable ingress and delay:} low-loss MZI meshes (LNOI/SiN) for fanout/selection; EO delay lines and phase shifters to meet $|\Delta t|\le W$; AWG/microring WDM for mode-parallelism.
% \item \emph{BSAs (port-bindable):} integrated interferometers (path/polarization/time-bin), SNSPD arrays with FPGA/ASIC coincidence logic, epoch-based binding control.
% \item \emph{Calibration/freeze:} pilot tones and dithers for phase/delay calibration; bindings updated every \(T_{\mathrm{freeze}}\) slots; only fine-trim loops run intra-epoch.
% \end{itemize}

\subsection{Throughput Region}
\label{ssec:throughput_EGS}

The expected number of end-to-end entanglement pairs for every pair $i,j$ given $x_{ij}$ is  $p^{\mathrm{e2e}}(\beta)\,x_{ij}$. We are interested in the \emph{achievable} long term per slot service rates $\lambda_{ij}$ for every pair $i,j$ that can be achieved with these possible schedules. Collect these in a vector $\lambda = (\lambda_{ij})_{i<j}$ and define the \emph{throughput region} $\Lambda_{\mathrm{EGS}}$ as the set of all such achievable vectors.

Recall that the convex set defined by inequalities
\eqref{eq:X_set_definition_budget_source_constraints} is denoted as
$\mathcal{X}$. It is convenient to interpret $\mathcal{X}$ in
graph-theoretic terms. Consider the complete graph
$G = (\mathcal{N},\mathcal{F})$ with node set $\mathcal{N}$ and edge
set $\mathcal{F} = \{(i,j): i<j\}$. For any integer vector
$x = \{x_{ij}\}_{(i,j)\in\mathcal{F}} \in \mathbb{Z}_{\ge 0}^{F}$,
the entry $x_{ij}$ specifies how many units of capacity
(or, equivalently, how many BSAs) are allocated to edge (or, equivalently pair)
$(i,j)$. Constraints \eqref{eq:X_set_definition_budget_source_constraints}
then enforce that (i) the total number of such units does not exceed
the budget $B$, and (ii) for each node $i$, the total
number of units incident on $i$ does not exceed  $S_i$. In combinatorial-optimization terminology, any such $x \in \mathcal{X}$ is a capacitated $b$-matching
of the complete graph with vertex capacities $b_i = S_i$ and an
additional cardinality constraint of $B$ units. The convex hull
$\mathrm{co}(\mathcal{X})$ can thus be viewed as a capacitated
$b$-matching polytope augmented by this global budget
constraint (e.g., \cite{schrijver2003combinatorial}).

%-----------------------------------------

%Recall that the convex set defined by inequalities \eqref{eq:X_set_definition_budget_source_constraints} is denoted as $\mathcal{X}$. 

Let $co(\cdot)$ denote the convex hull of the set $\mathcal{X}$, then the set $\Lambda_{\mathrm{EGS}}$ is given by ${{\Lambda_{\mathrm{EGS}} \equiv co(\mathcal{X}) \cdot p^{\mathrm{e2e}}(\beta)/\Delta t}}$. Therefore, the throughput region is given by \vspace{-3pt}
\begin{align}
\label{eq:capacity-region}
\Lambda_{\mathrm{EGS}} =   
\bigg\{  
& \frac{p^{\mathrm{e2e}}(\beta)}{\Delta t} \cdot \lambda: \lambda =\sum_{x \in \mathcal{X}} \delta_{x} x, \notag \\ 
&  \sum_{x \in \mathcal{X}} \delta_{x} = 1,  \ \delta_{x} \ge 0,  \quad \textup{for all } x \in \mathcal{X} \bigg\}.  
\end{align}
The set $\Lambda_{\mathrm{EGS}}$ contains all possible end-to-end entanglement generation rate vectors that an EGS can achieve. We dropped the dependency with $\beta$ for ease of notation.

\textbf{Maximum admissible total throughput:}
$\Lambda_{\mathrm{EGS}}$ is a high-dimensional polytope in $\mathbb{R}^{F}$ induced by the feasible station allocations $\mathcal{X}$ (constraints~\eqref{eq:X_set_definition_budget_source_constraints}). To quantify the “size’’ of $\Lambda_{\mathrm{EGS}}$ in terms of total throughput, we consider the maximum aggregate end-to-end rate per second:
\[
R^{\mathrm{sec}}_{\mathrm{EGS}}
\;\triangleq\;
\max_{\lambda\in\Lambda_{\mathrm{EGS}}}\ \sum_{(i,j) \in \mathcal{F}}\lambda_{ij}
\;=\;
\max_{\lambda\,\in\,\mathrm{co}(\mathcal{X})}
\;\frac{p_{\mathrm{e2e}}(\beta)}{\Delta t}\sum_{(i,j) \in \mathcal{F}}\lambda_{ij}.
\]
In other words, $R^{\mathrm{sec}}_{\mathrm{EGS}}$ represents the maximum L1 norm vector in the  throughput region. Clearly, the limitation of this metric is that it does not give a holistic picture of the region. 
\ifextension
However, to further motivate the usage of $R^{\mathrm{sec}}_{\mathrm{EGS}}$ as a metric of the size of the throughput region, we introduce the following remark that connects it with the maximum achievable symmetric service rate (fair service rate allocation).

\begin{remark}
\label{remark:sum_throughtput_symmetric_load_connection}
Let $\kappa \in \mathbb{R}$ and $\mathbf{1} \in \mathbb{R}^{F}$ be the all-ones vector. Assume that the switch is homogeneous across node pairs. If $\kappa^\star \triangleq \max \left\{ \kappa : \kappa \mathbf{1}  \in \Lambda_{\mathrm{EGS}} \right\}.
$ Then $\kappa^\star = R^{\mathrm{sec}}_{\mathrm{EGS}} / F$.
\end{remark}

We prove Remark~\ref{remark:sum_throughtput_symmetric_load_connection} in the Appendix (Sec.~\ref{sec:appendix_remark_symmetric}). This remark shows that the \emph{maximum total throughput} and the \emph{maximum achievable symmetric rate} are equivalent in our framework (up to a normalization) and provide intuitive measure of the throughput region. 
\else
\fi
The main theorem of the section follows and provides a closed form formula for $R^{\mathrm{sec}}_{\mathrm{EGS}}$.

\begin{theorem}[Closed form for maximum total throughput]
\label{thm:egs-sum-throughput}
For the EGS model with $B$ BSAs and per-node multiplexing degrees $\{S_i\}$, the maximum expected aggregate end-to-end success rate is
\[
R^{\mathrm{sec}}_{\mathrm{EGS}} \;=\; \frac{p_{\mathrm{e2e}}(\beta)}{\Delta t}\,E^{\max}_{\mathrm{EGS}}, \quad \text{where}
\]
\begin{equation}
\label{eq:Emax}
E^{\max}_{\mathrm{EGS}}
\;=\;
\min\bigg\{
B,\;
\Big\lfloor \tfrac{1}{2}\sum_{i} S_i \Big\rfloor,\;
\sum_{i} S_i \;-\; S_{\max}
\bigg\},
\end{equation}
and $S_{\max}\triangleq \max_{i} S_i.$
\end{theorem}
We now present a lemma that gives intuition in $E^{\max}_{\mathrm{EGS}}$, and is useful for the proof of  Theorem~\ref{thm:egs-sum-throughput} as well as for the rest of the analysis.

%[Cardinality of an optimal capacitated b-matching]
\begin{lemma}
\label{lem:bmatching}
Let $\mathcal{X}$ be the integral polytope with ${x=\{x_{ij}\}_{(i,j)\in \mathcal{F}}}$ obeying
inequalities \eqref{eq:X_set_definition_budget_source_constraints}.
% $\sum_{i<j}x_{ij}\le B$ and $\sum_{j\neq i}x_{ij}\le S_i$ for all $i$.
Then the maximum cardinality of a feasible allocation is $E^{\max}_{\mathrm{EGS}}$ in~\eqref{eq:Emax}:  \vspace{-4pt}
\[
\max_{x\,\in\,\mathcal{X}}\;\sum_{(i,j)\in \mathcal{F}} x_{ij} \;=\; E^{\max}_{\mathrm{EGS}}.  \vspace{-5pt}
\]
\end{lemma}

Both proofs of Lemma~\ref{lem:bmatching} and Theorem~\ref{thm:egs-sum-throughput} are in the Appendix (Sec.~\ref{sec:appendix_lemma_cardinality} and~\ref{sec:appendix_theorem_total_throughput}).

\ifjournal
\begin{remark}
    If we do not have the constraint of the sources (constraint \eqref{eq:source-budget}), and we are only limited by the number of BSAs, the throughput region becomes $\Lambda_{\mathrm{EGS}} = \{ p_{\mathrm{e2e}}(\beta) f_{\mathrm{EGS}} \lambda :  \sum_{i<j} \lambda_{ij} = B \}$. This can be proved since the convex hull of the new truncated $\mathcal{X}$ is the scaled simplex. Note that the total throughput in this case is constant (no matter the rate vector direction). \panos{Me, this remark should either be proved or moved to the journal version.}
\end{remark}
\else
\fi

\subsection{End-to-End Entanglement Fidelity}

In addition to the throughput analysis, we now derive a closed-form expression for the
fidelity of the end-to-end entangled state shared between any pair $i,j$ following Sec.~\ref{ssec:prelim_LLE_model}.

Recall that, given an LLEG protocol, $w_{i}(\beta)$
denotes the Werner parameter associated with a successful LLE of node $i$ (LLEG and link-level heralding in the switch related noise). Therefore, following Sec.~\ref{subsec:werner_states}, the Werner parameter of the end-to-end entanglement between nodes $i,j$ is
$$
w_{\text{EGS}}^{\text{e2e}}(\beta) \triangleq w_{i}(\beta)\,w_{j}(\beta)\,q_{\mathrm{BSM}} \ e^{-2 L/(v_f T)},
$$
where $L$ denotes the link length and $v_f$ the speed of light in the link. The last term, $e^{-2 L/(v_f T)}$, corresponds to the Werner parameter related to the memory decoherence of the two end nodes' entangled qubits from the time the switch learns about whether the LLEGs were successful until the end-to-end entanglement is heralded. This expression assumes that the end-to-end pair is used by the clients only after the end-to-end entanglement is heralded. Alternative models are possible (e.g., prepare-and-measure with postselection), in which case the relevant timing/decoherence window would be modified.

For demonstration, we further specialize this expression by assuming an LLEG protocol with midpoint sources on each link,
which generate pairs with initial Werner parameter $w_0(\beta)$ (a decreasing function of $\beta$). 
% Recall that $\tau_{\mathrm{hrld}}$ denotes the overhead time needed after the successful LLE has been established to herald it in the switch. Since we place the source in the middle of the link, after an established LLE we need to transmit a classical bit from the node to the switch for the link-level herald. 
Recall that $\tau_{\mathrm{hrld}}$ denotes the classical-notification latency associated with link-level heralding, i.e., the time it takes for the switch to learn whether the client--switch link attempt succeeded. In the EGS model the switch does not wait for this information before attempting a BSM; rather, BSMs are triggered blindly and the subsequently received link-level heralds are used to determine whether that BSM attempt was valid (both links succeeded) or wasted (at most one link succeeded). Since we place the source in the middle of the link, after a successful link attempt a classical bit is transmitted from the node to the switch to report the link-level herald. Hence, $\tau_{\mathrm{hrld}} = L/v_f$ and $w_{i}(\beta) = \,w_{j}(\beta) = w_0(\beta) e^{-\tau_{\text{hrld}}/T}$ for every pair $i,j$ since, in this architecture, only the state in the nodes is stored in a memory while we herald the LLE (see Sec.~\ref{ssec:EGS_system_formulation}).

The corresponding end-to-end fidelity of the swapped pair is $F_{\mathrm{EGS}}^{\mathrm{e2e}}(\beta)= \frac{1+3w_{\mathrm{EGS}}^{\mathrm{e2e}}(\beta)}{4}.$ Note that the end-to-end fidelity $F_{\mathrm{EGS}}^{\mathrm{e2e}}(\beta)$ scales with $\beta$ as the square of $w_{0}(\beta)$.

% 3. Use your "if" switch
\ifjournal
\subsection{Dynamics and Throughput-Optimal Policy}
\label{ssec:throughput-optimal}
\panos{This subsection will be relevant in the journal version. Probably will not use it in the conference. It serves on 1) proving that the throughput region is what we said earlier, and 2) introducing the max weight policy that will be helpful later to compare the whole regions/polytopes and not just the total maximum thoughput.}

We consider the slotted model of Section~\ref{sec:model_EGS} with epoch-based reconfiguration: control decisions (port bindings and station allocations) are chosen once every calibration epoch of $T_{\text{freeze}}$ slots and held fixed intra-epoch; see Secs.~\S2.2--\S2.3 for the feasible set $\mathcal{X}$ defined by \eqref{eq:X_set_definition_budget_source_constraints} and the throughput region $\Lambda_{\mathrm{EGS}} = f_{\mathrm{EGS}} \,p_{\text{e2e}}(\beta)\cdot \mathrm{co}(\mathcal{X})$ in \eqref{eq:capacity-region}. 
\paragraph{Queueing model and arrivals.}
For each unordered client pair $(i,j)$ we maintain a backlog (request queue) $Q_{ij}(t)$, updated at slot $t$ as
\[
Q_{ij}(t{+}1)=\bigl[\,Q_{ij}(t)-\mu_{ij}(t)\,\bigr]^+ + A_{ij}(t),
\]
where $A_{ij}(t)$ are exogenous requests with rates $\lambda_{ij}=\mathbb{E}[A_{ij}(t)]$ (finite second moments), and $\mu_{ij}(t)$ is the service offered in slot $t$. If the epoch-$e$ allocation is $x_{ij}(e)\in\mathcal{X}$ (held for $t\in\{eT_{\text{freeze}},\dots,(e{+}1)T_{\text{freeze}}-1\}$), then the \emph{expected} service per slot equals
\[
\mathbb{E}[\mu_{ij}(t)\mid x(e)] \;=\; \,p_{\text{e2e}}(\beta)\,x_{ij}(e),
\]
with $p_{\text{e2e}}(\beta)$ given in Sec.~\S2.3 and for $e = t \ mod \ T_{freeze}$. Hence, in this model, we have to optimize our decisions $x_{ij}(1), \dots$ for every pair of nodes $i, j$ to satisfy all the requests or maximize the throughput depending on what the goal of the switch is.
\paragraph{Max-Weight (MW) policy at epoch boundaries.}
At the start of epoch $e$, observe backlogs $Q_{ij}(eT_{\text{freeze}})$ and choose
\begin{equation}
\label{eq:MW}
x^{\star}(e)\in\arg\max_{x\in\mathcal{X}} \sum_{i<j} Q_{ij}(eT_{freeze})\,x_{ij}.
\end{equation}
Because $p_{\text{e2e}}(\beta)$ is a positive constant multiplier for all pairs in this model, maximizing the expected weighted service $\sum_{i<j}Q_{ij}(e)\,p_{\text{e2e}}(\beta)\,x_{ij}$ is equivalent to \eqref{eq:MW}. The selected $x^\star(e)$ is then held for all $T_{\text{freeze}}$ slots of the epoch (bindings are calibrated once per epoch). The theorem below is standard in the literature, we added the calibration periods restriction.
\begin{theorem}[Throughput optimality of epoch-based Max-Weight]
\label{thm:MW}
Suppose the arrival rate vector $\boldsymbol{\lambda}=(\lambda_{ij})$ lies in the interior of the throughput region, $\boldsymbol{\lambda}\in\mathrm{int}(\Lambda_{\mathrm{EGS}})$. Then the epoch-based Max-Weight policy \eqref{eq:MW} stabilizes the system (all queues are positive recurrent with finite time-average backlog), and hence achieves the full throughput region $\Lambda$.
\end{theorem}
\begin{proof}    
Let $L(\mathbf{Q})=\tfrac12\sum_{i<j}Q_{ij}^2$ be a quadratic Lyapunov function. Consider the $T_{\text{freeze}}$-slot \emph{frame drift}
\[
\Delta_T \;\triangleq\; \mathbb{E}\!\left[L(\mathbf{Q}(t{+}T_{\text{freeze}}))-L(\mathbf{Q}(t)) \mid \mathbf{Q}(t)\right].
\]
Standard queueing algebra yields (e.g., \cite{neely2010stochastic}, Ch.~4) the bound
\[
\Delta_T \;\le\; C\,T_{\text{freeze}} \;+\; \sum_{i<j} Q_{ij}(t)\,\Big( \lambda_{ij}\,T_{\text{freeze}} \;-\; \mathbb{E}\!\big[\textstyle\sum_{\tau=0}^{T_{\text{freeze}}-1}\mu_{ij}(t{+}\tau)\mid \mathbf{Q}(t)\big] \Big),
\]
for some finite $C$ (from second-moment bounds). Because the epoch action is constant, $\mathbb{E}\big[\sum_{\tau=0}^{T_{\text{freeze}}-1}\mu_{ij}(t{+}\tau)\mid \mathbf{Q}(t)\big]=T_{\text{freeze}}\,f_{\mathrm{EGS}}\,p_{\text{e2e}}(\beta)\,x_{ij}$, so the drift upper bound becomes
\[
\Delta_T \;\le\; C\,T_{\text{freeze}} \;+\; T_{\text{freeze}}\,\sum_{i<j} Q_{ij}(t)\,\bigl(\lambda_{ij}-f_{\mathrm{EGS}}\,p_{\text{e2e}}(\beta)\,x_{ij}\bigr).
\]
Given $\boldsymbol{\lambda}\in\mathrm{int}(\Lambda_{\mathrm{EGS}})$, there exists $\epsilon>0$ and a (possibly randomized) stationary policy $\bar{x}\in\mathrm{co}(\mathcal{X})$ with $f_{\mathrm{EGS}}\,p_{\text{e2e}}(\beta)\,\bar{x}_{ij}\ge \lambda_{ij}+\epsilon$ for all $(i,j)$ (Slater condition). Choosing $x=x^\star(e)$ via \eqref{eq:MW} minimizes the right-hand side over $\mathcal{X}$ each epoch and thus yields
\[
\Delta_T \;\le\; C\,T_{\text{freeze}} \;-\; \epsilon\,T_{\text{freeze}}\,\sum_{i<j} Q_{ij}(t),
\]
which implies negative drift outside a bounded set and therefore mean-rate stability and positive recurrence (see \cite{tassiulas1990stability}). This is the classic Max-Weight throughput-optimality argument of Tassiulas and Ephremides \cite{tassiulas1990stability}, extended to \emph{frame-based} (epoch) decisions; the frame only scales constants and does not shrink the stability region as long as $T_{\text{freeze}}<\infty$ and channel statistics are stationary within each epoch.
In other words, we apply a frame-based Lyapunov argument on the embedded chain at epoch boundaries; establishing negative expected drift at those sampling times suffices for stability, so Max-Weight chosen once per epoch is throughput-optimal for any finite $T_{freeze}$.
\end{proof}
Although the paper is not for designing throughput optimal control policies, this max weight notion will appear later when we will compare the throughput regions.
\paragraph{Remarks.}
(i) The MW objective \eqref{eq:MW} is a capacitated maximum $b$-matching with weights $Q_{ij}(e)$ over $\mathcal{X}$ (constraints \eqref{eq:X_set_definition_budget_source_constraints}). Efficient solvers (e.g., network flow formulations) can compute $x^\star(e)$ exactly. 
(ii) The constant calibration/freeze overhead does \emph{not} affect throughput optimality: decisions are taken once per epoch (every “round” of time slots), and MW remains optimal for the region $\Lambda$ derived in \eqref{eq:capacity-region}. 
\vspace{0.25em}
\noindent\textbf{Practical note.} If desired, one may replace $Q_{ij}(t)$ by any nondecreasing function thereof (e.g., age- or deadline-weighted backlogs); the classical MW proof still applies provided the weights dominate queues in the large-$Q$ regime (see \cite{neely2010stochastic}, Ch.~6).
\else
\fi

\ifextension
\subsection{Finite Quantum Memories in the Nodes}
\label{ssec:EGS_finite_node_memories}

In the previous, we assumed that each node has
sufficiently many local memories to buffer all tentative qubit halves while the
BSM outcomes are classically heralded, so that the next slot can start
immediately after the in-switch BSMs. We now relax this assumption and
consider the \emph{minimal-memory} regime where each node $i$ has only
$S_i$ memories---just enough to latch and hold at most one qubit per
multiplexing frame in a slot. In this case, a node cannot reuse a memory for a
new LLEG attempt until it learns whether the previous qubit was consumed in a
successful end-to-end pair (or discarded), i.e., until the \emph{end-to-end}
herald arrives. Hence, this case can be seen as a \emph{lower bound} of the performance in terms of throughput while the infinite node-memory assumption acts as an \emph{upper bound}. Note, that although this assumption is not directly relevant to the switch's hardware it affects its minimum slot duration (hence its throughput) and thus we have to distinct between the cases.

Under this assumption, the time slot should encapsulate the full time duration from the LLEG attempt until the end-to-end herald. Let $\Delta t_{\mathrm{EGS-lb}}$ denote the time slot duration of the EGS architecture under this assumption, then\footnote{The term $\Delta t$ can be in reality shorter than a full LLEG attempt period. Indeed, because the sources are idle for part of each time slot, they can start the first emission before the slot starts.
% , so part of the LLE attempt period is effectively overlapped with the previous slot.
}
\[
\Delta t_{\mathrm{EGS-lb}} = \Delta_t + \tau_{\mathrm{hrld}} + L/c. 
\]

Operationally, steps (i)--(iii) in the EGS slot description remain unchanged,
but step (iv) can no longer overlap with the next slot.

All per-slot success probabilities and fidelity calculations remain the same. The only change is the conversion of the per-second switch throughput due to the longer cycle time. Hence, the throughput region becomes the same
polytope scaling as in \eqref{eq:capacity-region}, with $\Delta t$ replaced by
$\Delta t_{\mathrm{EGS-lb}}$:
\vspace{-4pt}
\begin{equation}
\Lambda_{\mathrm{EGS-lb}}
\;=\;
\frac{p^{\mathrm{e2e}}(\beta)}{\Delta t_{\mathrm{EGS-lb}}}\;
\mathrm{co}(\mathcal{X}).
\label{eq:capacity-region-finite}
\end{equation}
Moreover, Theorem~\ref{thm:egs-sum-throughput} carries over verbatim with
this replacement:
\[
R_{\mathrm{EGS-lb}}^{\mathrm{sec}}
\;=\;
\frac{p_{\mathrm{e2e}}(\beta)}{\Delta t_{\mathrm{EGS-lb}}}\,
E^{\max}_{\mathrm{EGS}},
\]
where $E^{\max}_{\mathrm{EGS}}$ is still given by \eqref{eq:Emax}. Thus, finite
node memories do not change the \emph{combinatorial} resource-allocation
structure of the EGS; they only reduce all per-second rates by the duty factor
$\Delta t / \Delta t_{\mathrm{EGS-lb}}$.

If nodes have more than $S_i$ memories (but not infinitely many), one can partially pipeline across
multiple slots while awaiting end-to-end heralds; this yields an intermediate
regime between Sec.~\ref{ssec:EGS_system_formulation} (effectively unbounded
buffering) and the fully synchronous model above. We focus on the minimal case with 
$S_i$ memories to keep the EGS comparison analytically transparent.
\else
\fi

\section{Logical Model 2: Quantum Switch with Internal Memories (Herald-Then-Swap Control)}
\label{sec:model2}

% \subsection{Motivation and Role}
% \label{sec:motiv-role}

This section models and analyzes a near-term, memory-equipped switch and quantifies the benefit of \emph{minimal} storage: storing photonic qubits in a memory only long enough for the switch to learn the slot’s connectivity and make informed swap decisions. Similar models can be found in \cite{vasantam2022throughput, tillman2024calculating, valls2023capacity, promponas2023full}. 
% This connectivity-aware control reduces wasted BSM attempts compared to the EGS, but it increases decoherence due to additional storage time and  lengthens the slot. 
This connectivity-aware control reduces wasted BSM attempts compared to the EGS and avoids the probabilistic optical-BSA bottleneck of EGS, but it increases decoherence due to additional storage time and lengthens the slot.

% We focus on this single-slot storage regime and do not consider longer-term buffering for distillation---its analysis is deferred to future work within the same framework.

% By conditioning on the connectivity, the switch avoids BSM attempts on empty quantum memories, possibly improving the end-to-end entanglement rate.
% On the other hand, it increases the time-slot duration and therefore (i) possibly reduces utilization, as each LLE protocol attempts once per slot, and (ii) inevitably increases decoherence due to longer qubit storage. This trade-off is highlighted in our analysis and is studied throughout the rest of the paper. We do \emph{not} consider longer storage for distillation---its analysis is deferred to future work within the same framework. 
\subsection{System Formulation}
\label{sec:sys-form}

The switch can perform \(B\) BSMs per slot and has \(M\ge2B\) quantum memories that store single qubit states for one slot. The switch dedicates \(M_i=S_i\) quantum memories to each client $i$ (so that \(\sum_i M_i = M\)): each multiplexing frame of the LLEG protocol has a dedicated switch memory on its ingress channel.\footnote{If $M_i\neq S_i$, only $\min\{M_i,S_i\}$ LLEs can be realized per slot; equivalently, in our analysis one can replace $M_i$ by $\min\{M_i,S_i\}$ (the bottleneck).} 
% \lleprelim{we can say that this is also in line with the emitter protocol so we also assume that for this reason}. 
This simplifies the exposition and makes this model a direct generalization of the EGS: each LLE can be \emph{briefly} buffered at the switch until connectivity is known. Similar to the EGS architecture model, we assume that user nodes have sufficiently many memories to buffer all tentative Bell-pair halves until the classical BSM outcomes are heralded (i.e., during the end-to-end entanglement herald). Under this assumption, the next slot can start immediately after the BSMs.

\textbf{Timing and slot-level operation.} Let $\Delta t_{\mathrm{mem}}$ denote the duration of a time slot. For demonstration of this duration, we assume that each multiplexing frame corresponds to a mid-link source that attempts to emit one photon pair per pulse and each node memory tries to latch it.  $\Delta t_{\mathrm{mem}}$ (and the associated qubit storage time) is determined by four serial latencies:  
(i) the duration of one LLEG attempt, $\Delta t$,\footnote{This $\Delta t$ can be  shorter than the duration of a full LLEG attempt because emissions can be \emph{pipelined} across slots. In particular, during the portion of a slot when the switch is busy with classical-processing/decision overhead, sources need not be idle: they may emit up to one one-way propagation time earlier so that the corresponding photons arrive at the start of the next slot.} (ii) the classical overhead link-level heralding time, $\tau_{\mathrm{hrld}}$,  (iii) write-flag+link-level heralds acquisition, $\tau_{c}$, and, iv) control \& actuation, $\tau_{a}$:
\vspace{-5pt}
\[
\Delta t_{\mathrm{mem}} \;=\; \Delta t + \tau_{\mathrm{hrld}} + \tau_{c} + \tau_{a}. \vspace{-5pt}
\]
In that case, $\tau_{\mathrm{hrld}} \;=\; \frac{L}{v_f}$, where \(L\) denotes the node–switch fiber length and \(v_f\) the speed of light through the physical link. Morever, $\Delta t = 1/f_{\mathrm{pulse}}$. An explicit extension to block-based attempt LLEG protocols is developed in Sec.~\ref{subsec:numerical_batched_attempts_find_optimal_K}.

$\tau_a > 0$ means that the switch needs time after it learns the connectivity to optimize the schedule. Hence, $\tau_a$ depends on the complexity of the scheduling algorithm.
In the EGS, we implicitly assumed that $\tau_a$ is negligible since we stick to a control decision for $T_{\mathrm{freeze}} \Delta t$ and do not need to change it in every time slot. While this is in principle possible in this model, it would require precomputing a distinct control action for every connectivity pattern, which is not scalable since the number of configurations grows exponentially in the number of clients. Operation within a slot \(t\) proceeds as follows

    \emph{(i) LLEG Attempt:} Every LLEG protocol carries out up to $S_i$ attempts %\lleprelim{This jargon only applies to midpoint source}
    \emph{(ii) Connectivity acquisition:} The switch receives heralding messages and stores the learned connectivity in an internal data structure.
    \emph{(iii) Matching:} The controller computes and implements a memory-to-BSM mapping.
    \emph{(iv) BSM:} Up to \(B\) selected memory pairs are measured out by the end of the slot. Recall that we assume infinitely many memories in the nodes so classical heralding of these BSM outcomes can take place afterward and can overlap with the next slot. If node memories are scarce, this delay of end-to-end entanglement herald can instead be absorbed into a larger effective slot duration. For a graphical illustration of the memory-equipped switch as well as its time-slot operation see Fig.~\ref{fig:main_illustration_switches}(b) and~\ref{fig:main_illustration_switches}(d) respectively.

% While steps (ii)–(iv) execute, memories are “busy” and LLE protocols are idle. This is a structural drawback of quantum memory-equipped operation that will reflect as a reduced utilization of the physical links. However, the end-to-end entanglement rate might increase due to the conditioning on connectivity.

% \panos{As all of the similar models I know of, they do not take calibration times to change the BSAs. Is that ok since we wait for $\Delta t_{\mathrm{mem}}$ before we use them in every time slot?}

\textbf{LLEG success:}
% Let $\eta_{\mathrm{m}} \in[0,1]$ be the transmissivities (including coupling) from source to client and to switch, and let $g_{\mathrm{m}} \in(0,1]$ be the additional efficiencies capturing, for example, frequency conversions and wavelength selection. Then 
The LLEG success probability under the midpoint link source, homogeneous memories assumption is 
$p_i(\beta) = p_{\mathrm{pair}}(\beta)\,\eta_{m}^{2}\,g_m^{2}.$
% \[
% p_i(\beta)
% \;=\;
% p_{\mathrm{pair}}(\beta)\,\eta_{m}^{2}\,g_m^{2}.
% \]
% The value of $p_i(\beta)$ is different in the batch attempt LLE \cite{jones2016design}, as we detail in Sec.~\ref{subsec:numerical_batched_attempts_find_optimal_K}. 
This derivation is for the mid-link,
single-attempt LLEG described above; the block-based attempt LLEG of~\cite{jones2016design}
leads to a different expression, derived in
Sec.~\ref{subsec:numerical_batched_attempts_find_optimal_K}. Given \(p_i(\beta)\), the connectivity $C_i(t)$ in slot $t$ satisfies: 
\begin{equation}
C_i(t) \sim \mathrm{Binomial}\!\left(M_i,\, p_i(\beta)\right).\label{eq:binomial_C}
\end{equation}
For each client \(i\), the sequence \(\{C_i(t)\}_t\) is i.i.d.

\textbf{Control variables:}
After connectivity is learned, the switch observes \(C(t)=\{C_i(t)\}_{i\in\mathcal N}\) and chooses a \emph{matching of occupied memories across client indices}, i.e., $y_{ij}(t)\in \mathbb{Z}_{\ge 0}, (i,j) \in \mathcal{F},$
% \[
% y_{ij}(t)\in \mathbb{Z}_{\ge 0},\quad i<j,
% \]
subject to the per-node supply and per-station budget, $\sum_{j < i} y_{ji}(t) +  \sum_{j > i} y_{ij}(t)  \;\le\; C_i(t)\quad (\forall i),\quad \text{and} \quad
\sum_{i<j} y_{ij}(t) \;\le\; B.$
% \[
% \sum_{j\neq i} y_{ij}(t) \;\le\; C_i(t)\quad (\forall i),\qquad
% \sum_{i<j} y_{ij}(t) \;\le\; B.
% \]
% Each unit of \(y_{ij}(t)\) routes two occupied memories (one tagged by \(i\), one by \(j\)) to a BSM.

\textbf{Swap success per attempt:}
Conditional on scheduling \(y_{ij}(t)\) swaps for pair \((i,j)\), each attempt succeeds with $p_{\mathrm{swap}}$. In the case of BSAs, $p_{\mathrm{swap}} \;=\; \xi^2\, p_{\mathrm{BSA}}$. However, since in this model BSMs are carried out between memory qubits and not necessarily through linear optics, swaps can be deterministic (i.e., $p_{\mathrm{swap}} = 1$) which is a benefit of memory-assisted switches.
% \[
% p_{\mathrm{swap}} \;=\; \xi^2\, p_{\mathrm{BSA}},
% \]
Hence, the number of  end-to-end Bell states produced for the client pair $(i,j)$ in slot \(t\) is $\mu_{ij}(t) \;\sim\; \mathrm{Binomial}\!\big(y_{ij}(t),\, p_{\mathrm{swap}}\big)$. 
% $(i,j) \in \cal{P}$.

\subsection{Throughput Region}
\label{sec:model2-throughput}
% \leo{After reading through most of this subsection, I can say that in general it could really use some additional explanations. It's very hard to follow.} \panos{better?}

Let $\mathcal N$ be the node set and $\mathcal F\triangleq\{(i,j): i<j,\ i,j\in\mathcal N\}$.
The \emph{connectivity state space} is $\mathcal{C} \triangleq \prod_{i\in\mathcal N}\{0,1,\ldots,M_i\}$. A vector $c\in\mathcal{C}$ is a connectivity state. The random variable $C = C(t) = \{C_i(t)\}_{i\in\mathcal N}$ samples a  connectivity state following \eqref{eq:binomial_C}. Since for every client $i$ the sequence $\{ C_i(t) \}_t $ is i.i.d., we can omit $t$ when the time slot is not relevant. The probability mass function factorizes as
\[
\pi(c)
\;\triangleq\;
\mathbb P\{C=c\}
\;=\;
\prod_{i\in\mathcal N}
\binom{M_i}{c_i}\,p_i(\beta)^{\,c_i}\big(1-p_i(\beta)\big)^{M_i-c_i}.
\]
Given a state $c$, the switch learns the available \emph{supply} $c_i$ of memories that have successfully established LLEs with node $i$ and may schedule up to $B$ BSMs. Let
{
\small
\[
\mathcal Y(c)
\triangleq
\Big\{\, y\in\mathbb Z_{\ge 0}^{F}\ :
\sum_{j < i} y_{ji} + \sum_{j > i} y_{ij} \le c_i\ , \forall i,
\sum_{(i,j)\in\mathcal F} y_{ij}\le B
\Big\}
\]
}
be the set of capacitated b-matchings  in state $c$, with a global budget constraint of $B$ units (see Sec.~\ref{ssec:throughput_EGS}). A vector $y \in \mathcal Y(c)$ is a valid schedule of BSMs \emph{given connectivity $c$}. Its convex hull, $co\,\mathcal Y(c)  \subset \ \mathbb R_{\ge 0}^{F},$ captures all randomized time-sharing schedules conditioned on $c$.

A (possibly randomized) stationary, causal scheduling rule is a map 
\(v:\mathcal C\to \mathsf{co}\,\mathcal Y(c)\) that, for each observed state \(c\), selects a convex combination of feasible matchings; equivalently, \(v(c)\) is the expected pairing vector in that slot under some randomization over \(\mathcal Y(c)\).

We are interested in the set of all achievable long-term end-to-end entanglement rate vectors $\lambda$, which we collect in the throughput region $\Lambda_{\mathrm{mem}}$ (in end-to-end pairs per second):
\begin{equation}
\begin{aligned}[t]
\Lambda_{\mathrm{mem}}
=\Big\{ &
\lambda\in\mathbb R_{\ge 0}^{F}:\; 
\exists\, v:\mathcal C\!\to\! \mathsf{co}\,\mathcal Y(\cdot)\ 
\\& \text{s.t.}\;
\lambda \;=\; \frac{p_{\mathrm{swap}}}{\Delta t_{\mathrm{mem}}}
   \sum_{c\in\mathcal C} \pi(c)\, v(c) \Big\}.
\end{aligned}
\label{eq:lambda-mem-mixture}
\end{equation}

\textbf{Maximum admissible total throughput:}
Let $E_{\mathrm{mem}}^{\max}(c)$ denote the maximum number of swaps in state $c$, i.e.,
\[
E^{\max}_{\mathrm{mem}}(c) \triangleq \max_{y\,\in\,\mathcal Y(c)}\;\sum_{(i,j)\in \mathcal{F}} y_{ij}.  
\]
To find a closed form for $E_{\mathrm{mem}}^{\max}(c)$, we use Lemma~\ref{lem:bmatching} by replacing $S_i$ with $c_i$ for every node $i$,
\[
\begin{aligned}
E_{\mathrm{mem}}^{\max}(c)
&= \min\!\bigg\{\, B,\ \Big\lfloor \tfrac12 \sum_i c_i \Big\rfloor,\ \sum_i c_i - \max_i c_i \bigg\}.
\end{aligned}
\]

The maximum \emph{expected} total end-to-end rate is then
\begin{align*}
   R_{\mathrm{mem}}^{\mathrm{sec}}
&\triangleq
   \max_{\lambda\in\Lambda_{\mathrm{mem}}}\ \sum_{(i,j) \in \mathcal{F}}\lambda_{ij} \\
&=
   \frac{p_{\mathrm{swap}}}{\Delta t_{\mathrm{mem}}}\;
   \max_{v(\cdot)}\;
   \sum_{c\in\mathcal C} \pi(c)\,
   \sum_{(i,j)\in\mathcal F} v_{ij}(c) \\
&\text{s.t.}\quad v(c)\in \mathrm{co}\,\mathcal Y(c)\quad\forall c\in\mathcal C.
\end{align*}
Because the maximization decouples over $c$, this is
\begin{align*}
   R_{\mathrm{mem}}^{\mathrm{sec}}
&=
   \frac{p_{\mathrm{swap}}}{\Delta t_{\mathrm{mem}}}\;
   \sum_{c\in\mathcal C} \pi(c)\,
   \max_{y\in \mathrm{co}\,\mathcal Y(c)}\;
   \sum_{(i,j)\in\mathcal F} y_{ij}.
\end{align*}
Since the objective is linear, the maximum over $\mathrm{co}\,\mathcal Y(c)$ is attained at an extreme point, so
\[
   \max_{y\in \mathrm{co}\,\mathcal Y(c)}\sum_{(i,j)\in\mathcal F} y_{ij}
   =
   \max_{y\in \mathcal Y(c)}\sum_{(i,j)\in\mathcal F} y_{ij}
   \;\triangleq\; E_{\mathrm{mem}}^{\max}(c),
\]
and therefore
\[
   R_{\mathrm{mem}}^{\mathrm{sec}}
   =
   \frac{p_{\mathrm{swap}}}{\Delta t_{\mathrm{mem}}}\;
   \sum_{c\in\mathcal C} \pi(c)\,E_{\mathrm{mem}}^{\max}(c).
\]

\ifextension
Following the same arguments used in Remark~\ref{remark:sum_throughtput_symmetric_load_connection}, we can prove that $\max \left\{ \kappa : \kappa \mathbf{1} \in \Lambda_{\mathrm{mem}} \right\} = R_{\mathrm{mem}}^{\mathrm{sec}} / F$. Hence, $R_{\mathrm{mem}}^{\mathrm{sec}}$ simultaneously captures (i) the maximum aggregate rate over all pairs and (ii)---up to normalization—the largest \emph{symmetric} service rate that the switch can offer to every user pair when traffic is homogeneous. Viewing the region through these two equivalent scalars is useful: the total rate summarizes the overall “size’’ of $\Lambda_{\mathrm{mem}}$, while the symmetric rate is more directly interpretable in fairness-oriented or uniformly loaded scenarios. However, obtaining
\else
Obtaining
\fi
a closed-form expression for the maximum expected total rate in this architecture is significantly harder than in EGS (cf. Theorem~\ref{thm:egs-sum-throughput}). Nevertheless, since $R_{\mathrm{mem}}^{\mathrm{sec}}$ is an expectation, we can estimate it through Monte Carlo by sampling connectivities $c$.

% The expectation in Eq.\,\eqref{eq:max_total_throughput_mem} can be estimated via . 

\ifjournal
or approximated analytically using the theorem below (assuming that we have sufficient BSAs for simplicity).
% Theorem statement (no proof)
\begin{theorem}
Let $M\triangleq\sum_{i} M_i$ denote the total number of memories and $p\triangleq p_i(\beta)$ the per-memory link success probability. Order the memory budgets as $M_{1}\ge M_{2}\ge\cdots$ and define the fractions $\varphi_k\triangleq M_{k}/M$. Then a piecewise approximation for the expected per-slot maximum number of successful pairings is
\[
\mathbb{E}\!\left[E^{\max}_{\mathrm{mem}}(C)\right]\approx
\begin{cases}
\displaystyle \frac{Mp}{2}-\frac{1}{2}\sqrt{\frac{Mp(1-p)}{2\pi}}, 
& \text{if a single client is dominant: }\varphi_1\approx \tfrac12,\\[8pt]
\displaystyle \frac{Mp}{2}-\frac{1}{2}\sqrt{\frac{2(\varphi_1+\varphi_2)\,Mp(1-p)}{\pi}}, 
& \text{if two clients are dominant: }\varphi_1\approx\varphi_2\approx \tfrac12,\\[8pt]
\displaystyle \frac{Mp}{2}, 
& \text{otherwise.}
\end{cases}
\]
\end{theorem}
\begin{proof}
    The proof is very large - probably will be in a technical report.
\end{proof}
The theorem above assumes an abundance of Bell-state–measurement (BSM) stations, i.e., the rate is never station-limited. If at most $B$ BSMs can be performed per slot, impose the constraint by truncating the expectation appropriately.
An important consequence of the theorem is architectural: a repeater-like configuration in which two nodes hold most memories (or, in the switch model, two users dominate) uses the memory pool inefficiently (in terms of achieved entanglement per memory) and yields a reduced expected total throughput. To utilize memories effectively in a network, it is preferable to balance the memories across users so the switch pairs entanglements broadly; in this regime the attainable throughput approaches $\min\{B,\,Mp/2\}$ rather than suffering the dominance-induced penalties.
\else
\fi

% ------------

% From Eq.\eqref{eq:max_total_throughput_mem}, notice that a closed form for the total maximum expected rate is more challenging to get than in the corresponding formula of Model 1 (see Eq.\eqref{eq:max_total_throughput_egs}). Hence, to quantify it we could estimate this expectation using Monte Carlo sampling. However, we can also use the theorem below.

% \begin{theorem}
%     \panos{Here I have some bounds/approximations for the $R_{\mathrm{mem}}^{\mathrm{sec}}$. I will include at as a theorem. This will help compare the total throughput for the different models.}
% \end{theorem}

\subsection{End-to-End Entanglement Fidelity}
\label{sec:model2-fidelity}

Recall that, given an LLEG protocol, $w_{i}(\beta)$
denotes the Werner parameter associated with a successful LLEG of node $i$ (LLEG and link-level heralding in the switch related noise). Therefore, the Werner parameter of the end-to-end entanglement between nodes $i,j$ is
$$
w_{\text{mem}}^{\text{e2e}}(\beta) \triangleq w_{i}(\beta)\,w_{j}(\beta)\,q_{\mathrm{BSM}} \ e^{-4(\tau_c + \tau_a )/T} \ e^{-2L/(v_f T)}.
$$
The term $e^{-4(\tau_c + \tau_a )/T} $ corresponds to the Werner parameter related to the memory decoherence that all of the four qubits involved experience after the link-level herald but before the BSM, whereas $e^{-2 L/(v_f T)}$ corresponds to the decoherence of the two end node qubits from the BSM until end-to-end entanglement is heralded. Hence, $F_{\mathrm{e2e}}^{\mathrm{mem}}(\beta) \;=\; \frac{1+3\,w_{\mathrm{mem}}^{\mathrm{e2e}}(\beta)}{4}.$

For demonstration, we again assume a mid-link source that attempts one photon-pair per optical pulse, with each node memory trying to latch the first arriving photon. Recall that $\tau_{\mathrm{hrld}}$ denotes the overhead time needed after the successful LLEG has been established to herald it in the switch. Since we place the source in the middle of the link, after an established LLE we need to transmit a classical bit from the node to the switch for the link-level herald in the switch. Hence, $\tau_{\mathrm{hrld}} = L/ v_f$ and $w_{i}(\beta) = \,w_{j}(\beta) = w_0(\beta) e^{-2 \tau_{\text{hrld}}/T}$ for every pair $i,j$ since now we have memories in both the node and the switch.

\ifjournal
\subsection{Dynamics and Throughput-Optimal Policy}
\label{sec:model2-dynamics}
\panos{This subsection will be relevant in the journal version. Probably will not use it in the conference. It serves on 1) proving that the throughput region is what we said earlier, and 2) introducing the max weight policy that will be helpful later to compare the whole regions/polytopes and not just the total maximum thoughput.}
We consider stochastic demands for end-to-end pairs and design a policy that is throughput-optimal for the region in \S\ref{sec:model2-throughput}.
\paragraph{State, queues, and dynamics.}
Similarly to EGS, let $A_{ij}(t)$ be the exogenous pair requests for $(i,j)\in\mathcal F$ arriving in slot $t$ with $\mathbb E[A_{ij}(t)] = \lambda_{ij}$ and finite second moments similarly with the previous model. Let $Q_{ij}(t)$ be the backlog for pair $(i,j)$ at the beginning of slot $t$. Per \S\ref{sec:sys-form}, the switch observes the connectivity state $C(t)$, chooses a matching $y(t)\in\mathcal M(C(t))$, and attempts swaps with success $\mu_{ij}(t)\sim \mathrm{Binomial}(y_{ij}(t),p_{\mathrm{swap}})$ independently across scheduled edges. The queue recursion is the same as in the previous model:
\begin{equation}
Q_{ij}(t+1) \;=\; \big[\,Q_{ij}(t) - \mu_{ij}(t)\,\big]^+ \;+\; A_{ij}(t),
\qquad (i,j)\in\mathcal F,
\label{eq:model2-queue}
\end{equation}
with slot duration $\Delta t$ and service rate per slot equal to $\mu_{ij}(t)$.
\paragraph{Max-Weight over per-state matchings.}
Define the per-slot weight for edge $(i,j)$ as $w_{ij}(t)\triangleqQ_{ij}(t)$. Given $C(t)$, select the matching $y^\star(t)$ that solves the capacitated \emph{maximum-weight} problem
\begin{equation}
y^\star(t)\ \in\ \arg\max_{y\in\mathcal M(C(t))}\ \sum_{(i,j)\in\mathcal F} w_{ij}(t)\, y_{ij}.
\label{eq:model2-maxweight}
\end{equation}
(Equivalently, maximize $\sum w_{ij}\,\mathbb E[\mu_{ij}(t)\mid y]=p_{\mathrm{swap}}\sum w_{ij}y_{ij}$; the constant factor $p_{\mathrm{swap}}$ does not affect the optimizer.) The problem \eqref{eq:model2-maxweight} is a capacitated $b$-matching with node budgets $C_i(t)$ and a global station budget $B$; it admits a polynomial-time solution via reductions to min-cost flow.
\paragraph{Throughput optimality.}
Let $\Lambda_{\mathrm{mem}}$ be the throughput region in \eqref{eq:lambda-mem-mixture}. Then the policy \eqref{eq:model2-maxweight} stabilizes the queues for any arrival rate vector $\lambda$ strictly interior to $\Lambda_{\mathrm{mem}}$ (i.e., $\lambda\in \mathrm{int}(\Lambda_{\mathrm{mem}})$). In particular, there exists $\epsilon>0$ such that if $\lambda+\epsilon\mathbf 1\in \Lambda_{\mathrm{mem}}$, then the quadratic Lyapunov drift under \eqref{eq:model2-maxweight} is negative outside a compact set, implying strong stability:
\[
\limsup_{T\to\infty}\frac1T\sum_{t=0}^{T-1}\sum_{(i,j)}\mathbb E\big[Q_{ij}(t)\big]\ <\ \infty.
\]
Hence \eqref{eq:model2-maxweight} is \emph{throughput-optimal} for Model~2. The proof of this is omitted since almost the same tools were employed in the models in \cite{vasantam2022throughput, tillman2024calculating, valls2023capacity}.
Similarly to Model 1, although we do not focus on the stability of the logical switch models, we will use the throughput optimal policies later for the comparison of the throughput region.
\paragraph{Remarks.}
(i) If the objective is sum throughput rather than queue stability (no exogenous $A_{ij}$), set $w_{ij}(t)\equiv 1$ to recover a per-state cardinality maximization that attains $E_{\max}(C(t))$ and achieves the sum rate. (ii) Weighted fairness (e.g., proportional fairness) can be enforced by choosing $w_{ij}(t)$ as suitable functions of the empirical service or virtual queues; the stability proof follows the same drift arguments.
\else
\fi

\section{Numerical Results}
\label{sec:numerics}

% \panos{I plan to have i) memory based specific analysis using only maximum total throughput, ii) compare the two models using only maximum total throughput as metric, iii) compare the whole regions of the two models. Below you can see (i) that I have so far. }

We consider a default scenario reflecting near-term hardware parameters.
Whenever we deviate from this set of parameters, we will state the changed parameter explicitly.

\textbf{Baseline scenario:}
% We consider \(N=6\) client nodes with multiplexing degree \(S_i=3\) and \(B=8\) BSM stations (BSAs for EGS). 
We consider $N=6$ client nodes with multiplexing degree $S_i=3$ and allow up to $B=8$ concurrent BSM operations per slot. 
% (implemented by optical BSAs in the EGS, and by the switch's memory/gate resources in the memory-equipped model).
For BSAs in the EGS model, detectors have efficiency \(\xi=0.90\) and \(p_{\mathrm{BSA}}=0.50\). For the memory assisted model we take $p_{\mathrm{swap}} = 1$. The fiber attenuation is \(\alpha=0.2~\mathrm{dB/km}\) over a node–switch distance \(L_{\text{node}\to\text{sw}}=1~\mathrm{km}\), so each half-link (\(0.5~\mathrm{km}\)) has transmittance \(\eta_{\mathrm{fiber},\,1/2}=10^{-\alpha (0.5)/10}=0.98\) with \(g_{sw}= g_m = 0.85\). For the LLEG protocol we use mid-link SPDC sources: we use the SPDC example of \cite{dhara2022heralded}, where the probability of generating an \(n\)-pair coincidence is \(p(n)=(n+1)\frac{\beta^n}{(\beta+1)^{n+2}}\). Under this model, the probability that the source emits photons is \(\Pr(n\ge 1)=\frac{\beta(\beta+2)}{(\beta+1)^2}\) and the resulting initial fidelity is \(F_{0}(\beta)=\frac{3w_{0}(\beta) + 1}{4} = \frac{\Pr(n=1)}{\Pr(n\ge 1)}=\frac{2}{(\beta+1)(\beta+2)}\). We set \(\beta=0.03\) to get $F_0(\beta) = 0.96$. We assume \(v_f=2.0\times10^8~\mathrm{m/s}\), yielding \(\tau_{\mathrm{hrld}}=5~\mu\mathrm{s}\), with \(\tau_c=2~\mu\mathrm{s}\), \(\tau_a=3~\mu\mathrm{s}\), and source pulse rate \(f_{\mathrm{pulse}}=10~\mathrm{MHz}\).  Using these SPDC sources, we assume that the LLEG protocol for the EGS architecture is the single attempt protocol described in Sec.~\ref{sec:model_EGS}, and for the memory-equipped architecture we use the optimal block-based LLEG protocol described in detail in Sec~\ref{subsec:numerical_batched_attempts_find_optimal_K}. For fidelity-related knobs we use \(T=500~\mu\mathrm{s}\) and \(q_{\mathrm{BSM}}=0.97\). We approximate the stochastic variables using Monte Carlo with \(10^5\) samples.

\subsection{Block-based LLEG Protocol in Memory-Equipped Models}
\label{subsec:numerical_batched_attempts_find_optimal_K}

In this section, we study the trade-offs between end-to-end entanglement rate,
end-to-end fidelity, and users’ utility in the memory-based model. We use our framework to show how a hardware designer can optimize the block size \(K\), in a block-based attempt LLEG protocol.
% , using the maximum total throughput \(R_{\mathrm{mem}}^{\mathrm{sec}}\). 
% derived with our framework. The block-based attempt LLEG protocol reduces to the single attempt protocol used for demonstration in Sec.~\ref{sec:model2} when $K=1$.

\textbf{Block-based attempt LLEG protocol:} When the photon source employed by the LLEG protocol supports a fast pulse rate $f_{\mathrm{pulse}}\gg1/\Delta t_{\mathrm{mem}}$, it is convenient to include $K$ consecutive LLEG attempts at the beginning of a slot, and to load only the first successful photon in the switch's quantum memory \cite{jones2016design}. Which photon is loaded is reconciled at the end of the block through classical heralding, and has to match between the two sides to have a success. No structural change is needed to incorporate this to our model, only a (block-size–dependent) adjustment to the qubit storage times, and an increased probability $p_{i}(\beta)$ for successful LLEG attempts. 
% and a midpoint source LLE protocol.

Assuming amid-link source, the time-slot duration is \vspace{-3pt}
\begin{equation}
\Delta t_{\text{mem}} \;=\;  {(K-1)}/{f_{\mathrm{pulse}}}  \;+\; \tau_{\text{herald}} \;+\; \tau_c \;+\; \tau_a. \label{eq:slot-time} 
\end{equation}
Larger \(K\) increases latency but raises the LLEG success probability. Success occurs iff the first non-empty bin in the block
delivers both photons. Let the per-bin outcomes be
$q_2 = p_{\text{pair}}(\eta g_m)^2, \ 
q_1 = p_{\text{pair}}\eta g_m (1-\eta g_m), \
q_0 = 1 - p_{\text{pair}}(2\eta g_m-(\eta g_m)^2)$ (i.e., both, exactly one and none of the photons succeeding respectively).
With \(K\) independent bins, the probability that the first non-empty outcome delivers both photons (i.e., the probability of a successful LLEG) is \vspace{-4pt}
\begin{equation*}
\begin{aligned}
p(\beta, K)  &= \sum_{t=1}^{K} q_0^{\,t-1}\,q_2 = q_2\,\frac{1-q_0^{K}}{1-q_0} \\
&= \frac{\eta g_m}{2-\eta g_m}\,\Bigl[1-\bigl(1-p_{\text{pair}}(\beta)  \cdot (2\eta g_m-(\eta g_m)^2)\bigr)^K\Bigr].
\end{aligned}
\label{eq:psucc}
\end{equation*}
%\begin{equation*}
%\begin{aligned}
%p(\beta, K) &=  \frac{\eta g}{2-\eta g}\,\Bigl[1-\bigl(1-p_{\text{pair}}(\beta)  \cdot (2\eta g-(\eta g)^2)\bigr)^K\Bigr].
%\end{aligned}
%\label{eq:psucc}
%\end{equation*}
The end-to-end entanglement fidelity also changes (from Sec.~\ref{sec:model2-fidelity}).  Specifically, what changes with $K$ is the time that the \emph{successful} qubits spend in memories before the link-level herald in the switch and, consequently, the amount of decoherence. Let $T_{\mathrm{succ}} \in \{1,\dots,K\}$ denote the index of the first bin in the block for which the outcome is ``non-empty and good'' (i.e., it yields the two photons that eventually define the link-level Bell pair). As in the success probability calculation above, the event
that the first non-empty bin is $t$ and delivers both photons has probability $\mathbb{P}\{T_{\mathrm{succ}} = t,\ \text{success}\} = q_0^{\,t-1}\,q_2,
\quad t=1,\dots,K,$ and the overall success probability is $p(\beta,K) = \sum_{t=1}^{K} q_0^{\,t-1} q_2.$ Hence, \vspace{-8pt}
\[
\mathbb{P}\{T_{\mathrm{succ}} = t \mid \text{success}\}
= \frac{q_0^{\,t-1} q_2}{p(\beta,K)},
\qquad t=1,\dots,K.
\]
% Therefore, the additional pre link-level herald storage time (beyond the single-attempt case) is approximately
% \begin{equation*}
% \begin{aligned}
% \delta\tau_{\mathrm{block}}(\beta,K)
% &\;\triangleq\;
% \mathbb{E}\!\left[(K-T_{\mathrm{succ}}) \Delta t
% \;\middle|\; \text{success}\right] \\
% &=
% \frac{\Delta t}{\,p(\beta,K)}
% \sum_{t=1}^{K} (K-t)\,q_0^{\,t-1} q_2.
% \end{aligned}
% \label{eq:delta-tau-batch}
% \end{equation*}
% For $K=1$ we have $\delta\tau_{\mathrm{block}}(\beta,1)=0$, so we recover the single-attempt model.

{\normalsize
Under the midpoint-source assumption, a successful LLEG on link $i$ corresponds to a mixture over possible success bins $t$, each inducing a different pre–link-level storage time. The Werner parameter associated with a successful LLEG on link $i$ in the block-based protocol is therefore}
{
\small
\[
w_i(\beta,K)
=
\frac{w_0(\beta)}{p(\beta,K)}
\sum_{t=1}^{K}
q_0^{\,t-1} q_2\,
\exp\!\left(
-\frac{2}{T}\Bigl(\tau_{\mathrm{hrld}} + (K-t)\Delta t\Bigr)
\right).
\]
}
Substituting this into the end-to-end expression (see Sec.~\ref{sec:model2-fidelity}), the Werner parameter of the shared pair between nodes $i,j$ under the block-based LLEG protocol becomes $w_{\mathrm{mem,block}}^{\mathrm{e2e}}(\beta,K)
\;\triangleq\;
w_i(\beta,K)\,w_j(\beta,K)\,
q_{\mathrm{BSM}}\,
e^{-4(\tau_c+\tau_a)/T}\,
e^{-2L/(v_f T)}.$
%\[
%w_{\mathrm{mem,bat}}^{\mathrm{e2e}}(\beta,K)
%\;\triangleq\;
%w_i(\beta,K)\,w_j(\beta,K)\,
%q_{\mathrm{BSM}}\,
%e^{-4(\tau_c+\tau_a)/T}\,
%e^{-2L/(cT)}.
%\]

% The corresponding end-to-end fidelity is
% $
% F_{\mathrm{mem,block}}^{\mathrm{e2e}}(\beta,K)
% \;=\;
% \frac{1 + 3\,w_{\mathrm{mem,block}}^{\mathrm{e2e}}(\beta,K)}{4}.
% $
% $F_{\mathrm{mem,block}}^{\mathrm{e2e}}(\beta,1)$ reduces to the non-blocked expression in Sec.~\ref{sec:model2-fidelity}.

% We now examine how the memory model behaves as we vary key parameters, using it to guide hardware/model choices.

\begin{figure*}[t]
  \centering
  \begin{minipage}{0.325\textwidth}
\centering\includegraphics[width=0.8\linewidth,height=0.15\textheight]{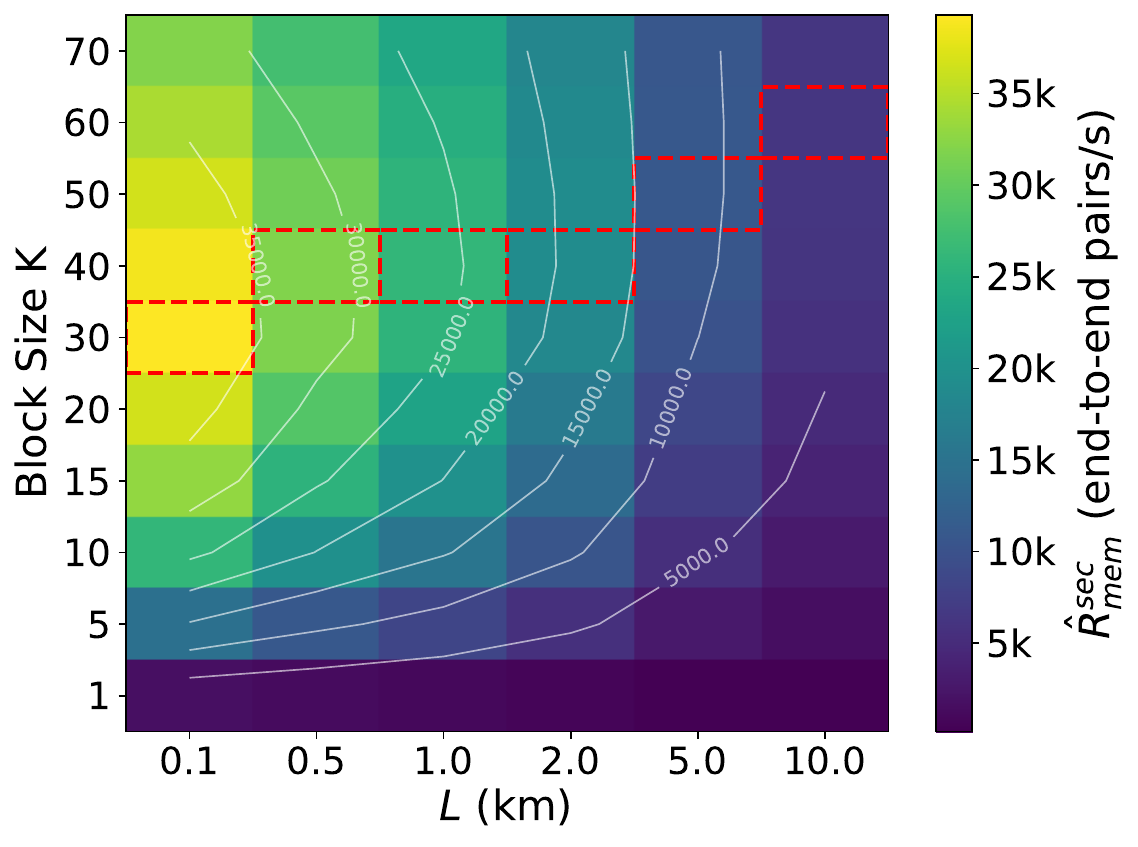}
    \vspace{-1.6ex}
\small (a)~$\beta = 0.03$ and $f_{\mathrm{pulse}}=10~\mathrm{MHz}$.
  \end{minipage}\hfill
  \begin{minipage}{0.325\textwidth}
    \centering
    \includegraphics[width=0.8\linewidth,height=0.15\textheight]{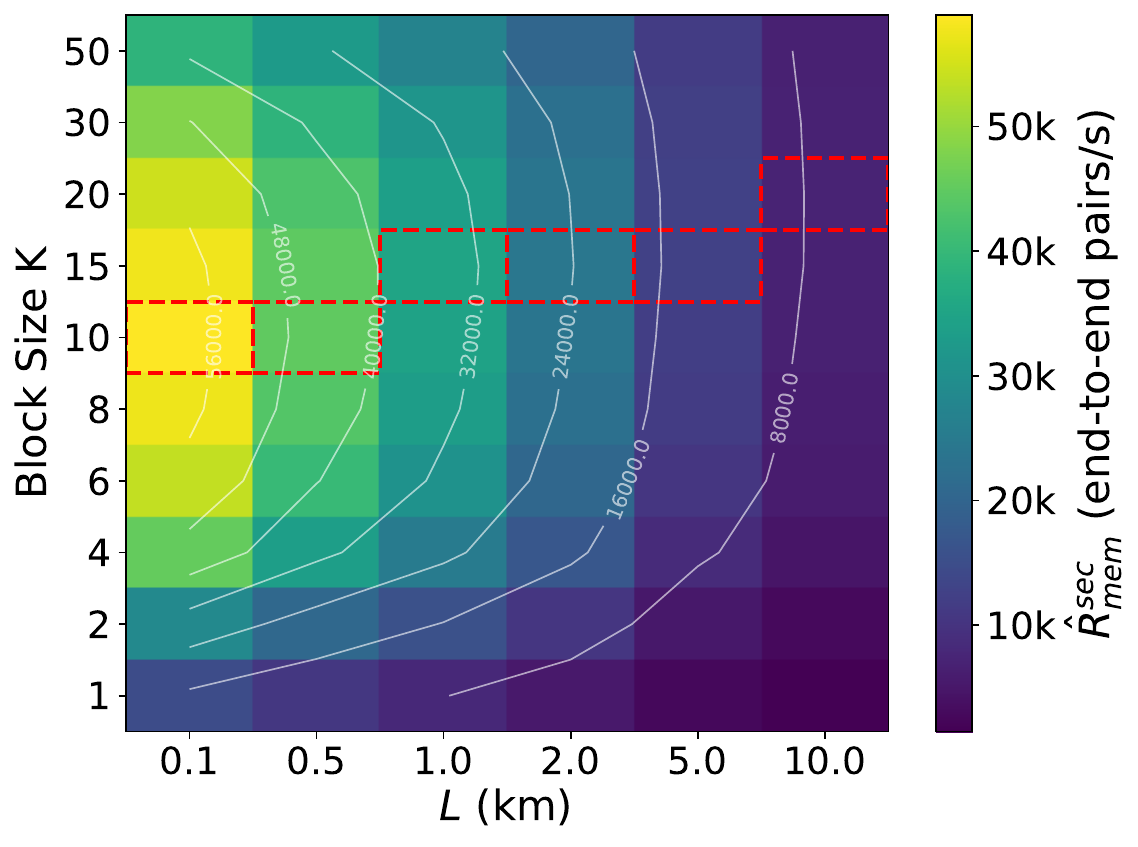} 
    \vspace{-1.6ex}
    \small (b)~\(\beta=0.15\) and $f_{\mathrm{pulse}}=10~\mathrm{MHz}$.
  \end{minipage}\hfill
  \begin{minipage}{0.325\textwidth}
    \centering
    \includegraphics[width=0.8\linewidth,height=0.15\textheight]{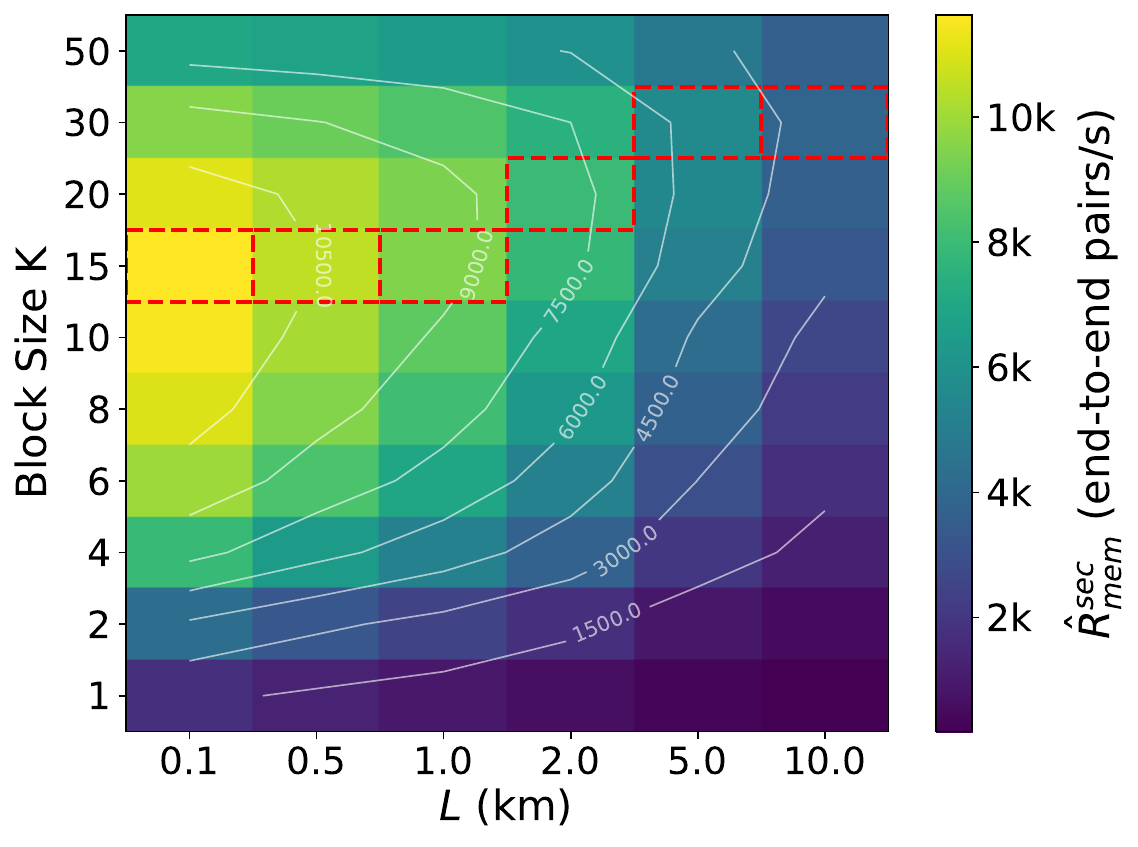} 
    \vspace{-1.6ex}
    \small (c)~$\beta = 0.03$ and \(f_{\mathrm{pulse}}=1~\mathrm{MHz}\).
  \end{minipage}
  \caption{Maximum total throughput \(\hat{R}_{\mathrm{mem}}^{\mathrm{sec}}\) as a function of block size \(K\) and link length \(L\); red dashed rectangles mark the  optimal \(K\) for each \(L\).}
  \label{fig:KL_heatmaps}
\end{figure*}
%here_check change R notation

\textbf{Maximum aggregated throughput versus \(K\) and \(L\):}
Fig.~\ref{fig:KL_heatmaps} contours how \(\hat{R}_{\mathrm{mem}}^{\mathrm{sec}} \triangleq R_{\mathrm{mem}}^{\mathrm{sec}} / F \) varies with \(K\) and distance \(L\). 
% We normalize the total rate with the number of node pairs to capture the service of end-to-end entanglements per flow. 
% \beta = 0.03, f_{\mathrm{pulse} = 10 MHz}
% \beta = 1, f_{\mathrm{pulse} = 10 MHz
% \beta = 0.03, f_{\mathrm{pulse} = 1 MHz}
In panel (a), under the default values $\beta = 0.03$ and $f_{\mathrm{pulse}}=10~\mathrm{MHz}$, a small \(L\) favors an optimal number of attempts of \(K\!=\!30\) (red dashed rectangles). As \(L\) increases, the link success probability drops, so a larger \(K\) is needed to maintain a high end-to-end rate; however, increasing \(K\) further becomes sub-optimal for the end-to-end rate because the slot duration grows per \eqref{eq:slot-time}, reducing the attempt rate. As \(L\) continues to grow, the rate declines and changing \(K\) yields diminishing returns. In panel (b), increasing \(\beta\) to~0.15 brightens the map and shifts the optimum to lower \(K\): higher single photon pair success lets us reduce \(K\), which shortens the slot and improves rate. In panel (c), slowing the sources (\(f_{\mathrm{pulse}}=1~\mathrm{MHz}\)) also lowers the optimal \(K\), since each added attempt increases the slot duration even more (see \eqref{eq:slot-time}).

\textbf{Fidelity heatmap versus \(K\) and source rate:}
We examine the end-to-end fidelity as a function of \(K\) and  \(f_{\mathrm{pulse}}\). The heatmap in Fig.~\ref{fig:fidelity_heatmap} shows that for fast sources, fidelity is insensitive to \(K\); so \(K\) can be chosen to optimize solely the rate. In contrast, for slower sources, increasing \(K\) lowers fidelity. The reason is that each additional attempt lengthens the slot (cf. Eq.~\eqref{eq:slot-time}), increasing memory storage time and thus decoherence. In general, raising \(K\) does not improve fidelity so the design problem is a rate–fidelity trade-off: increase \(K\) only if the rate gain justifies the fidelity loss.

%herefigure1
% \begin{figure}[t]
%   \centering
%   \includegraphics[width=0.54\linewidth]{Figures/mem_model_fidelity_contour/fidelity_heatmap_beta_0.2.pdf}
%   \vspace{-0.3cm}
%   \caption{End-to-end fidelity contours versus batch size \(K\) and source clock \(f_{\mathrm{pulse}}\) (other parameters at the default scenario). 
%   % For high \(f_{\mathrm{pulse}}\), fidelity is flat in \(K\); for low \(f_{\mathrm{pulse}}\), larger \(K\) degrades fidelity.
%   }
% \label{fig:fidelity_heatmap}
% \end{figure}

\begin{figure}[t]
  \centering
  \begin{minipage}[t]{0.49\linewidth}
    \centering
    \includegraphics[width=\linewidth]{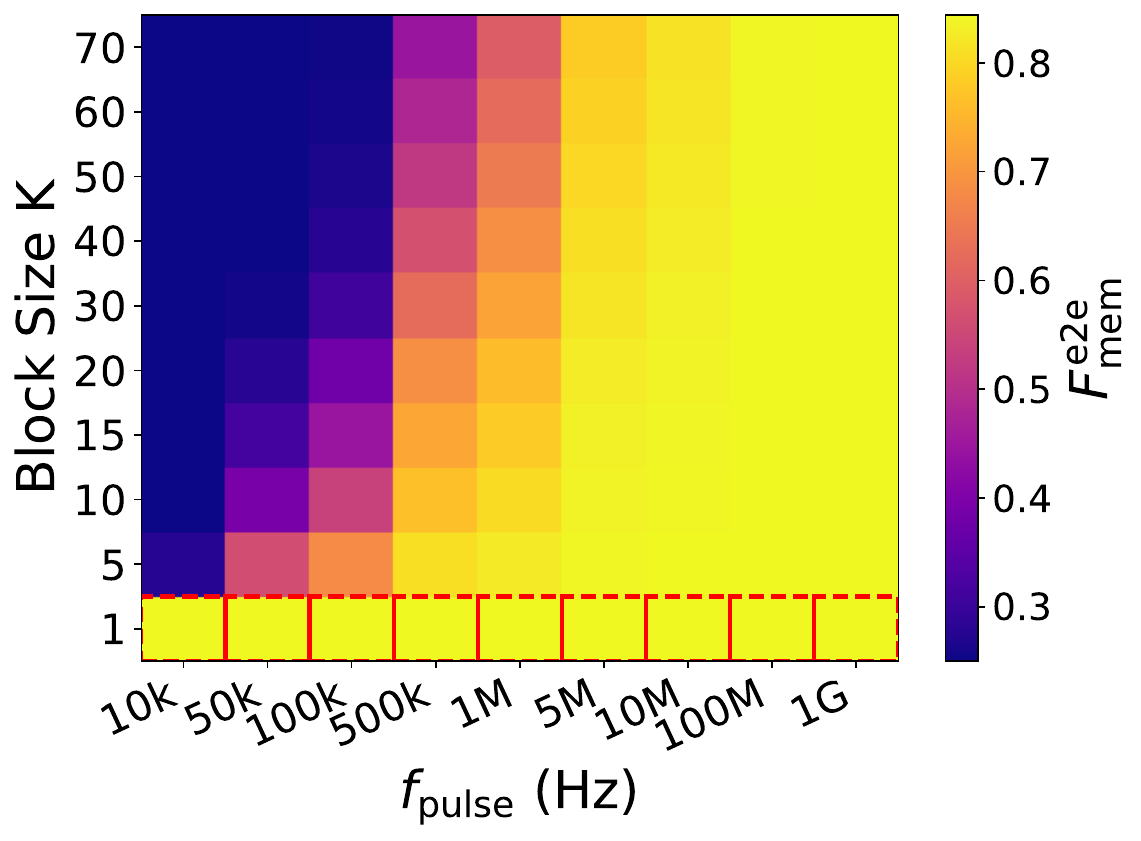}
    \vspace{-0.9cm}
    \captionof{figure}{End-to-end fidelity heatmap versus block size \(K\) and \(f_{\mathrm{pulse}}\).}
    \label{fig:fidelity_heatmap}
  \end{minipage}\hfill
  \begin{minipage}[t]{0.49\linewidth}
    \centering
    \includegraphics[width=\linewidth]{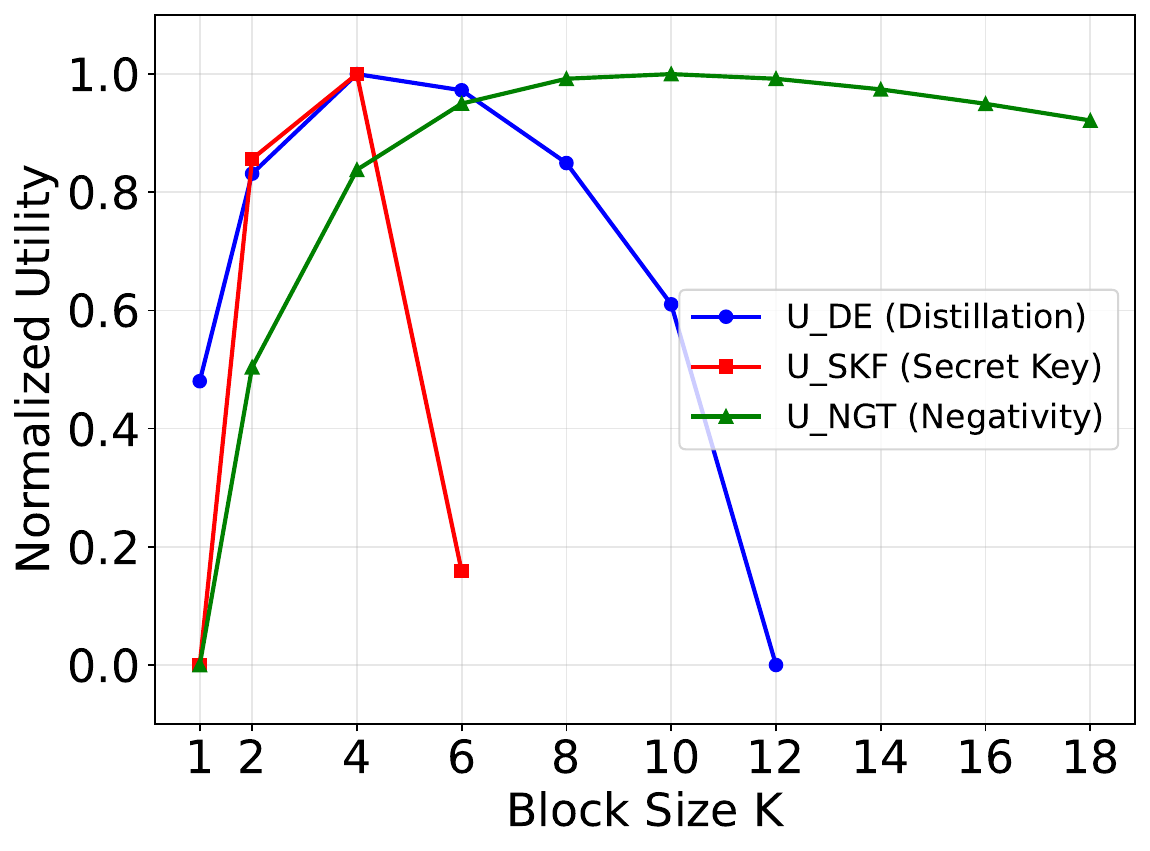}
    \vspace{-0.9cm}
    \captionof{figure}{Utilities versus block size \(K\).}
    \label{fig:utilities_vs_K}
  \end{minipage}
\end{figure}

%camera-ready
%remove here for length
% \textbf{End-to-end rate/fidelity trade-off via utilities:}
% To study the joint trade-off between end-to-end entanglement \emph{rate} \(R\) and \emph{fidelity} \(F\), we evaluate proportional-fair (log) utilities of the form \(U=\log\!\big(R\,Q(F)\big)\). We use three application-motivated choices for \(Q(F)\): (i) \emph{Distillation utility} \(U_{\mathrm{DE}}=\log\!\big(R\,D_H(F)\big)\) with hashing yield \(D_H(F)=1+F\log_2 F+(1-F)\log_2\!\frac{1-F}{3}\) (asymptotic one-way distillable entanglement); (ii) \emph{Secret-key utility} \(U_{\mathrm{SKF}}=\log\!\big(R\,S_{\mathrm{BB84}}(F)\big)\), where \(p=\tfrac{2}{3}(1-F)\), \(h(p)=-p\log_2 p-(1-p)\log_2(1-p)\), and \(S_{\mathrm{BB84}}(F)=\max\{1-2h(p),0\}\) (BB84 secret-key fraction); and (iii) \emph{Negativity utility} \(U_{\mathrm{NGT}}=\log\!\big(R\,\max\{F-\tfrac12,0\}\big)\) (scaled entanglement negativity above the separability threshold). These utilities \cite{vardoyan2023quantum} provide a single objective that balances throughput with quality and allows comparing hardware/model choices on equal footing. For $U_{DE}$ and $U_{SKF}$ to be nonnegative we need fidelities greater than $\sim 0.85$ and as observed in Fig.~\ref{fig:fidelity_heatmap}, the default scenario can give $F_{\mathrm{mem}}^{\mathrm{e2e}} \in [0.3, 0.8]$. For that reason, below we use $L = 0.1km$. 
\textbf{End-to-end rate/fidelity trade-off via utilities:}
To study the joint trade-off between end-to-end entanglement \emph{rate} \(R\) and \emph{fidelity} \(F\), we use proportional-fair utilities of the form \(U=\log\!\big(R\,Q(F)\big)\), following \cite{vardoyan2023quantum}. We consider three application-motivated choices for \(Q(F)\) and thus $U$: (i) \emph{distillable entanglement}, $U_{\mathrm{DE}}$, with \(Q(F)=D_H(F)\), the hashing yield (asymptotic one-way distillable entanglement); (ii) \emph{secret-key utility}, $U_{\mathrm{SKF}}$, with \(Q(F)=S_{\mathrm{BB84}}(F)\), the BB84 secret-key fraction; and (iii) \emph{negativity utility}, $U_{\mathrm{NGT}}$, with \(Q(F)=\max\{F-\tfrac12,0\}\), a scaled entanglement negativity above the separability threshold. These utilities provide a single objective that balances throughput with quality and allow a comparison of hardware/model choices on equal footing. Since \(U_{\mathrm{DE}}\) and \(U_{\mathrm{SKF}}\) are only relevant for high fidelities (\(F\gtrsim 0.85\)), while the default scenario yields \(F_{\mathrm{mem}}^{\mathrm{e2e}}\in[0.3,0.8]\) (Fig.~\ref{fig:fidelity_heatmap}), below we use \(L=0.1\,\mathrm{km}\).
% here figure2
% \begin{figure}[t]
%   \centering\includegraphics[width=0.7\linewidth]{Figures/mem_model_utilities/utilities_vs_batch_size_beta_0.05tohavelargefidelities.pdf}
%   \vspace{-0.3cm}
%   \caption{Utilities versus batch size \(K\). 
%   The optimal \(K^\star\) depends on the chosen utility (application) and hardware parameters.}
%   \label{fig:utilities_vs_K}
% \end{figure}

Fig.~\ref{fig:utilities_vs_K} plots \(U_{\mathrm{DE}}, U_{\mathrm{SKF}}, U_{\mathrm{NGT}}\) versus block size \(K\). For comparable y-axis, we normalize the values of the utilities from 0 to 1. All three exhibit a unimodal, concave-like dependence, yielding a best block size \(K^\star\). \(K^\star\) varies with the utility (i.e., target application) and the hardware parameters, enabling application-aware tuning of the parameters of the quantum switch. In the current setting, all three utilities preserve an optimal $K^\star \in [4,10]$.

\subsection{Comparing EGS and Memory Switching}
\label{sec:comparison}

% We compare the EGS model and the quantum memory-equipped switch (Models 1 and 2 respectively). Our goal is to identify which architecture \emph{dominates} as the parameter regime changes and to provide guidance for hardware-model co-design. We use two complementary lenses: (i) the \((R,F)\) \emph{fidelity–rate frontier} obtained by sweeping the loss parameter \(\beta\), and (ii) an application-oriented \emph{negativity utility} $U_{\mathrm{NGT}} \;=\; \log\!\big(R\,\max\{F-\tfrac{1}{2},\,0\}\big),$
% which compresses each operating point into a single score aligned with entanglement usefulness.

% Section~\ref{sssec:baseline_frontier_negativity} establishes the baseline dominance at the default operating point; Section~\ref{sssec:scenario_sensitivity} then varies key hardware knobs to map where each model is preferable.

In this section, we compare the EGS and the quantum memory-equipped switch. Our goal is to identify which architecture \emph{dominates} as the parameter regime changes and to provide guidance for hardware–model co-design. Although the results are specific to the chosen parameter regimes, our framework is general and can be applied to other parameters of interest for comparison. We use two complementary lenses: (i) the \((R,F)\) \emph{fidelity–rate frontier} obtained by sweeping the tuning parameter \(\beta\), and (ii) an application-oriented \emph{negativity utility} \(U_{\mathrm{NGT}}=\log\!\big(R\,\max\{F-\tfrac{1}{2},0\}\big)\), which compresses each operating point into a single score aligned with entanglement usefulness. Whenever the memory model is used, we choose the optimal block size \(K^\star\) for each setting— \(K^\star(\beta)\in \arg\max_K U_{\mathrm{NGT}}(K,\beta)\). Thus \(K^\star\) is case dependent.

% Section~\ref{sssec:baseline_frontier_negativity} establishes the baseline dominance at the default operating point; Section~\ref{sssec:scenario_sensitivity} then varies key hardware knobs to study the robustness of the findings.

\begin{figure}[t]
  \centering
  \begin{minipage}{0.5\columnwidth}
    \centering
    \includegraphics[width=\linewidth]{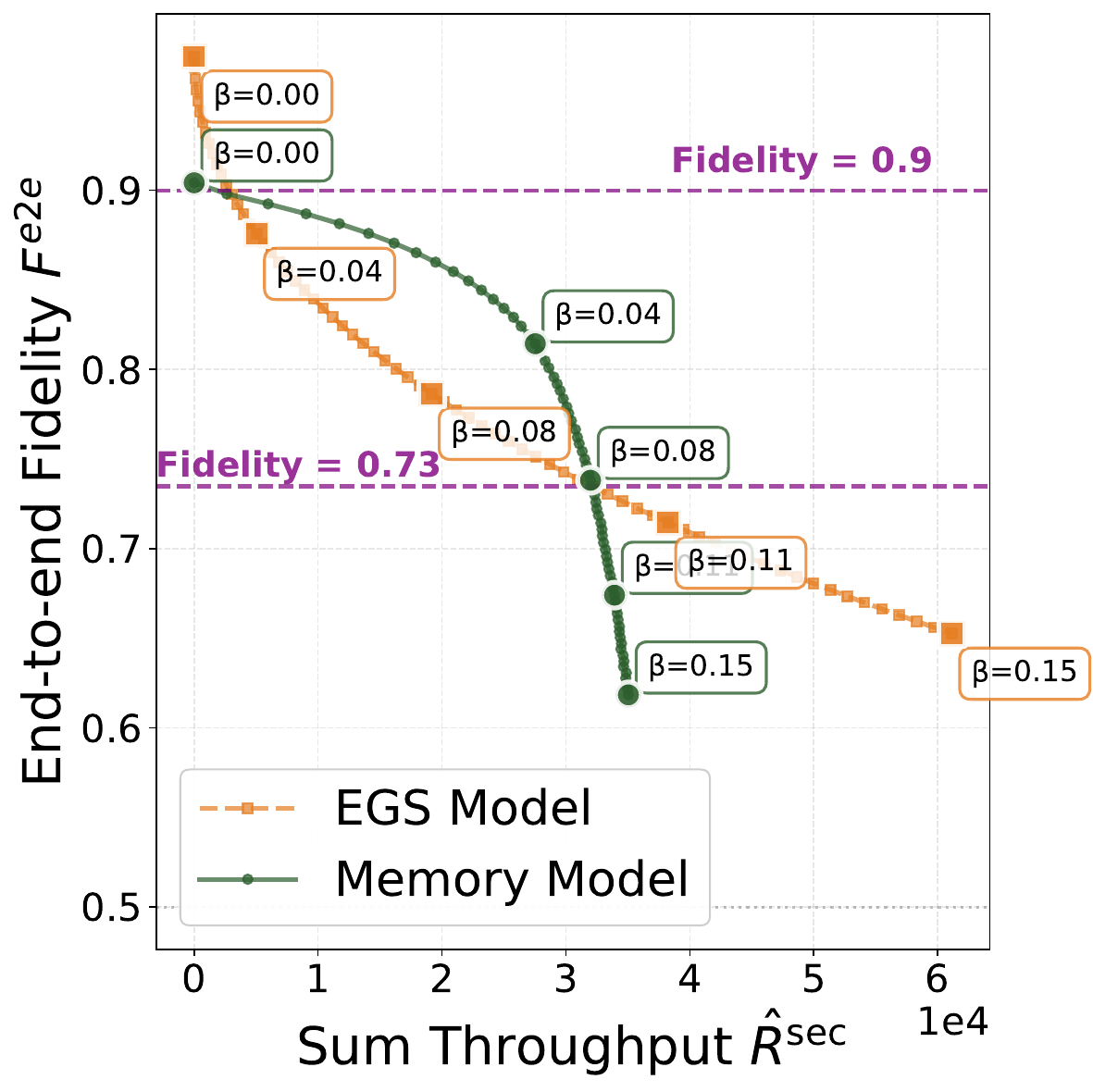}
    \small (a) Fidelity–rate frontiers w.r.t $\beta$.
  \end{minipage}\hfill
  \begin{minipage}{0.5\columnwidth}
    \centering
    \includegraphics[width=\linewidth]{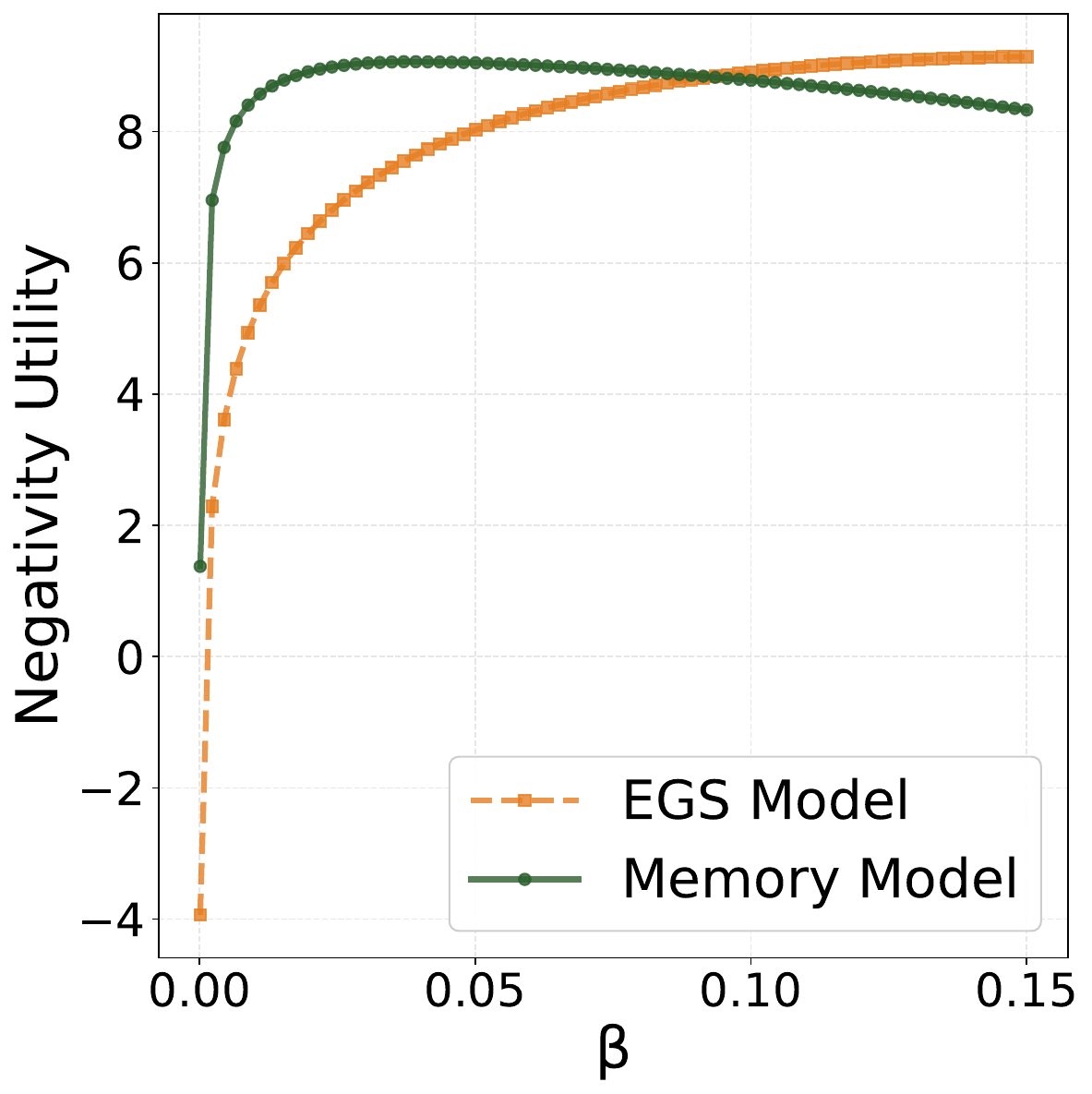}
    \small (b) Negativity utility vs.\ $\beta$.
  \end{minipage}
    \vspace{-0.4cm}
  \caption{Baseline comparison of EGS and memory-equipped switching.}
  \label{fig:rf_frontier}
\end{figure}

\subsubsection{Fidelity–Rate Frontier and Negativity vs.  $\beta$}
\label{sssec:baseline_frontier_negativity}

We sweep the parameter \(\beta \in [0,\,0.15]\) and we compute the achievable normalized end-to-end rate \(\hat{R}^{\mathrm{sec}}_{\cdot}\) and fidelity \(F^{\mathrm{e2e}}_{\cdot}\) for the EGS and quantum memory-equipped switches.

% \begin{figure}[t]
%   \centering
%   \begin{subfigure}[t]{0.24\textwidth}
%     \centering
%     \includegraphics[width=\linewidth]{Figures/sec2_a_baseline_fidelityrateandutilities/rate_fidelity_tradeoff.pdf}
%     \caption{Fidelity–rate frontiers as \(\beta\) is swept. Example labels mark selected \(\beta\) values; a horizontal reference at \(F=0.85\) is shown for reading off application thresholds.}
%     \label{fig:rf_frontier}
%   \end{subfigure}\hfill
%   \begin{subfigure}[t]{0.24\textwidth}
%     \centering
%     \includegraphics[width=\linewidth]{Figures/sec2_a_baseline_fidelityrateandutilities/negativity_utility_vs_beta.pdf}
%     \caption{Negativity utility \(U_{\mathrm{NGT}}=\log\!\big(R\,\max\{F-\tfrac12,0\}\big)\) versus \(\beta\) for both models.}
%     \label{fig:negativity_beta}
%   \end{subfigure}
%   \vspace{-0.5em}
%   \caption{Baseline comparison of EGS (Model~1) and quantum memory-equipped (Model~2) switching under the default scenario.}
% \end{figure}

\begin{figure*}[t]
  \centering
  % --- (a) ---
  \begin{minipage}{0.29\textwidth}
    \centering\includegraphics[width=0.9\linewidth,height=0.14\textheight]{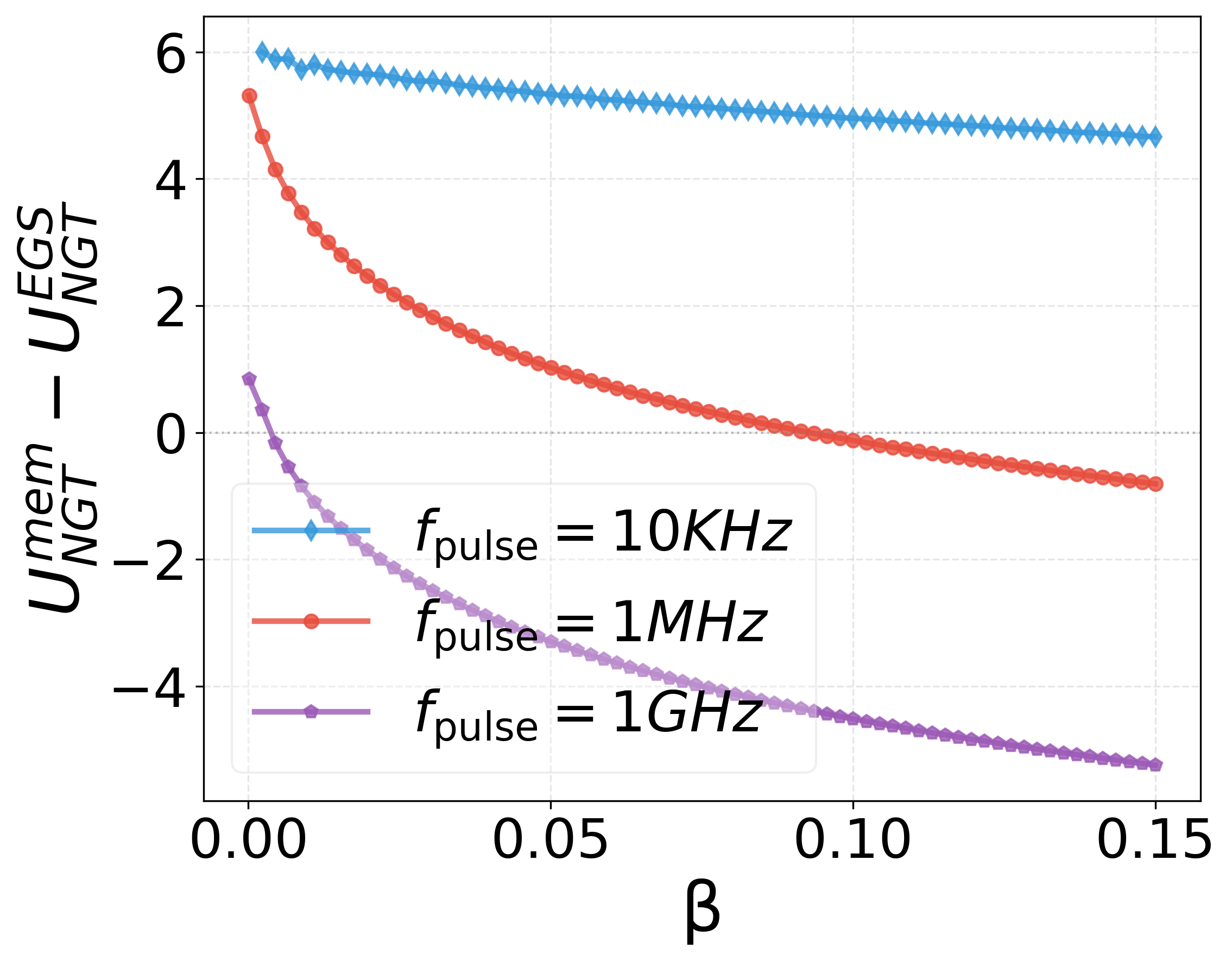} \vspace{-0.1pt}
    \vspace{-0.8ex}
    \small (a)~Varying source rate $f_{\mathrm{pulse}}$.
  \end{minipage}\hfill
  % --- (b) ---
  \begin{minipage}{0.29\textwidth}
    \centering\includegraphics[width=0.9\linewidth,height=0.14\textheight]{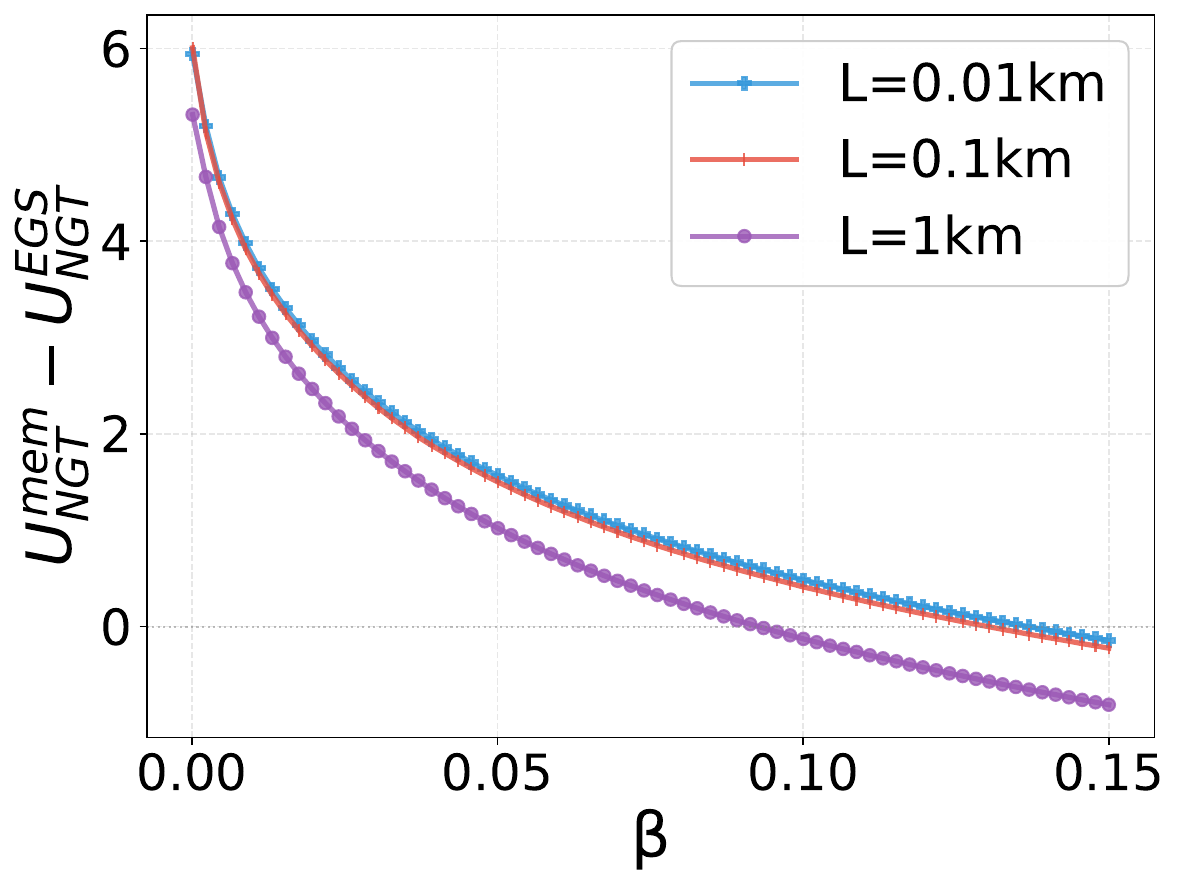} \vspace{-0.1pt}
    \vspace{-0.8ex}
    \small \qquad (b) Varying fiber length $L$.
  \end{minipage}\hfill
  % --- (c) ---
  \begin{minipage}{0.29\textwidth}
    \centering\includegraphics[width=0.9\linewidth,height=0.14\textheight]{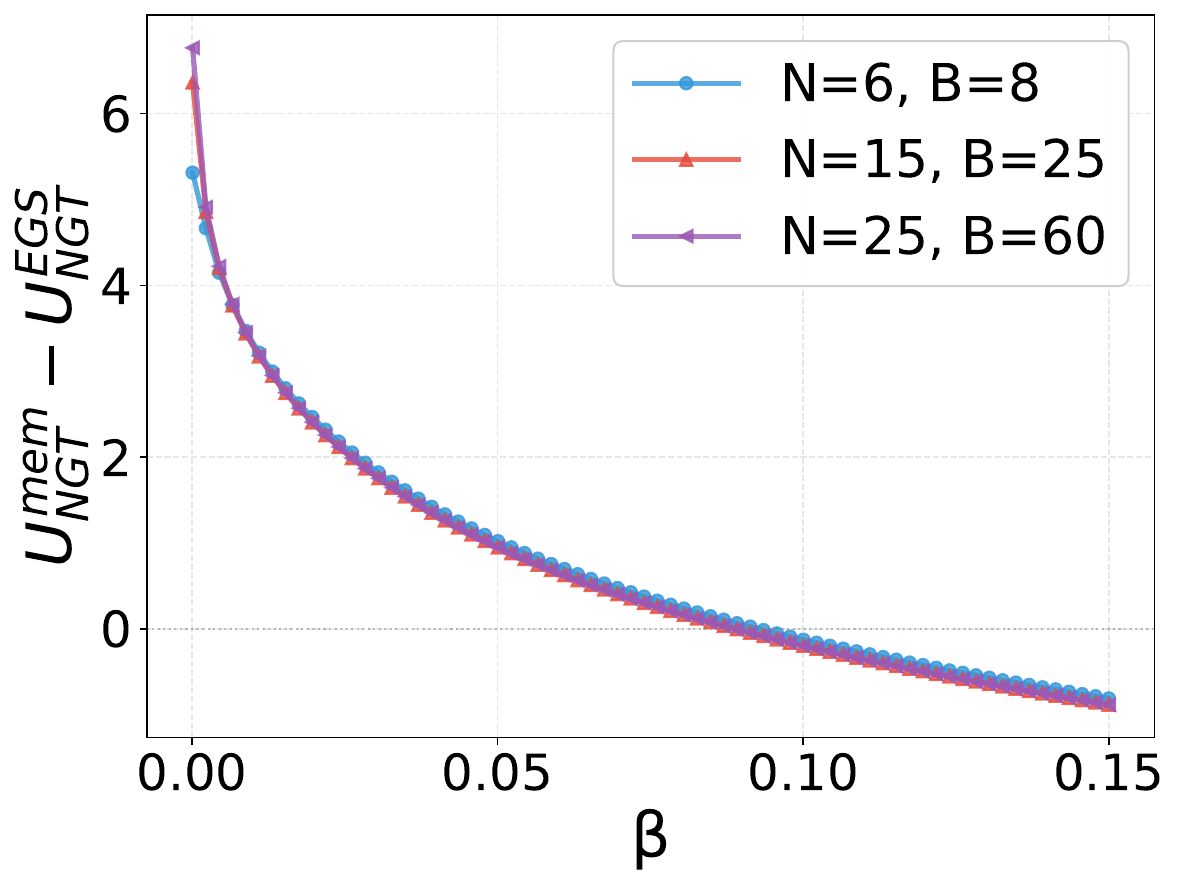} \vspace{-0.1pt}
    \vspace{-0.8ex}
    \small \quad (c) Scaling the switch size $(N,B)$.
  \end{minipage}
  \vspace{-0.6ex}
  \caption{Sensitivity of dominance using $\Delta U_{\mathrm{NGT}}(\beta)$. Positive values favor the memory model; negative values favor EGS.}
  \label{fig:panel_sensitivity}
\end{figure*}

In Fig.~\ref{fig:rf_frontier}(a) observe that the curves trace the set of \((R,F)\) versus \(\beta\).
Across large $\beta$, EGS achieves \emph{both} higher rate and higher fidelity. However, in an intermediate high-fidelity window (\(0.73 \lesssim F \lesssim 0.9\)), the memory model becomes preferable: for the same target fidelity it delivers a higher rate. For very high fidelities (\(F \gtrsim 0.9\)) the memory model cannot reach the EGS curve because decoherence in the switch memories caps its attainable fidelity under this parameter set. In practice, an application can place a horizontal line at its required fidelity and read the \emph{achievable rate} from Fig.~\ref{fig:rf_frontier}(a). 
% The memory model offers the better rate at fidelities in the range $[0.75,0.92]$, while EGS wins once the required \(F\) is relaxed or pushed to very high values. 
In Fig.~\ref{fig:rf_frontier}(b), the negativity utility shows that for large \(\beta\) (corresponding to lower fidelities in Fig.~\ref{fig:rf_frontier}(a)) EGS yields higher \(U_{\mathrm{NGT}}\), whereas for smaller \(\beta\) the memory model, with block size optimized for utility, attains higher \(U_{\mathrm{NGT}}\) in the high-fidelity regime.
% In Fig.~\ref{fig:rf_frontier}(b), note that the negativity utility leads to the same conclusion:
% there exists a \(\beta\)-region (lower loss) where the memory model has larger \(U_{\mathrm{NGT}}\), reflecting its advantage at useful fidelities.
% As \(\beta\) increases, EGS becomes better in \(U_{\mathrm{NGT}}\); however, this happens at \(\beta\) where the resulting fidelities are too low to be acceptable for many applications.
Together, the two plots indicate that \emph{model preference is regime-dependent}.

\subsubsection{Regime Sensitivity}
\label{sssec:scenario_sensitivity}

We vary key parameters to see how the ordering between EGS and the quantum memory-equipped switch changes. To keep a single, comparable score per operating point, we plot the \emph{difference in negativity utility} \(\Delta U_{\mathrm{NGT}} \triangleq U_{\mathrm{NGT}}^{\mathrm{mem}} - U_{\mathrm{NGT}}^{\mathrm{EGS}}\) across \(\beta\),
% other utilities could be substituted depending on the application.
and we indicate for each case, the fidelity interval over which the memory-equipped architecture dominates, inferred from the corresponding \((R,F)\) plot as in Fig.~\ref{fig:rf_frontier}(a). The \((R,F)\) plots are omitted for space.

\textbf{Source rate (Fig.~\ref{fig:panel_sensitivity}(a)).}
Relative to the baseline, \emph{slower} sources shift \(\Delta U_{\mathrm{NGT}}\) upward. The memory-equipped switch \emph{dominates} because the slot is no longer a bottleneck and connectivity learning boosts rate while keeping fidelity high; for \(f_{\mathrm{pulse}}=10\,\mathrm{kHz}\) this corresponds to dominance over \(0.65 \lesssim F \lesssim 0.92\). With \emph{fast} sources (e.g., \(1\,\mathrm{GHz}\)), \(\Delta U_{\mathrm{NGT}}<0\) across all \(\beta\) (EGS \emph{dominates}); the memory model wastes many entangled pairs awaiting for the heralding. However, this dominance assumes that nodes can store all states while awaiting heralding, which may be impractical at high $f_{\mathrm{pulse}}$.
% from the nodes while the EGS is boosted.

\textbf{Length (Fig.~\ref{fig:panel_sensitivity}(b)).}
Increasing $L$ favors the memory-equipped switch and the \(\beta\)-interval where it dominates widens. Shorter heralding paths reduce decoherence and effective slot duration, improving rate and fidelity for the memory design; in $L=\{0.1,0.01\}km$ it outperforms EGS for \(0.66 \lesssim F \lesssim 0.93\).

\textbf{Switch size (Fig.~\ref{fig:panel_sensitivity}(c)).}
For larger switches, the curves are largely consistent with the baseline, supporting robustness of the conclusions. 

% with a mild tilt toward EGS (slightly more negative \(\Delta U_{\mathrm{NGT}}\)) at larger switches.

%todo 2) model 2 with memory 3) model 3 with fiber delay lines 4) reallocate memories 5) comparison

%in delay lines can be the pump since many in delay lines?. 
%what if you can reallocate the memories every t_freeze depending on the state. question how many memories would we give to nodes? 

\section{Conclusion}

This work provides a unified framework to characterize the rate-fidelity tradeoffs between all-photonic and memory-equipped quantum switches. We demonstrate that while internal memories enable efficient herald-then-swap control, their advantage is regime-dependent, governed by the interplay of coherence times, heralding delays, and source rates. Future work should expand this analysis beyond near-term switches to include high-stability memories capable of supporting entanglement distillation and error correction, as well as exploring alternative switch architectures and LLEG protocols.

\section{Appendix}
\label{sec:appendix}

\ifextension
\subsection{Proof of Remark~\ref{remark:sum_throughtput_symmetric_load_connection}}
\label{sec:appendix_remark_symmetric}
To prove Remark~\ref{remark:sum_throughtput_symmetric_load_connection}, we use the following lemma.

\begin{lemma}
\label{lemma:symmetric_convex_set}
Let $d \in \mathbb{N}$ and let $\Lambda \subset \mathbb{R}^d$ be a non-empty, convex set such that $x \in \Lambda \ \Rightarrow\ P x \in \Lambda$ for every permutation matrix $P \in \mathbb{R}^{d \times d}$. Define $S^\star \;\triangleq\; \max_{x \in \Lambda} \sum_{k=1}^d x_k
\quad\text{and}\quad
\kappa^\star \;\triangleq\; \max\{\kappa \in \mathbb{R} : \kappa \mathbf{1} \in \Lambda\},$
where $\mathbf{1} \in \mathbb{R}^d$ is the all-ones vector. Assume $S^\star < +\infty$ and that the maximum is attained. Then
\[
\kappa^\star \;=\; \frac{S^\star}{d}.
\]
\end{lemma}

\begin{proof}
We prove the two inequalities $\kappa^\star \le S^\star/d$ and $\kappa^\star \ge S^\star/d$. First, we prove that $\kappa^\star \le S^\star/d$. Let $\kappa \in \mathbb{R}$ be such that $\kappa \mathbf{1} \in \Lambda$. Then $\sum_{k=1}^d (\kappa \mathbf{1})_k \;=\; d \kappa \;\le\; S^\star,$ by definition of $S^\star$. Hence $\kappa \le S^\star / d$ for every feasible $\kappa$, which implies
$\kappa^\star \le S^\star/d$.

Then, we prove that $\kappa^\star \ge S^\star/d$.
Let $x^\star \in \Lambda$ be a maximizer of $\sum_{k=1}^d x_k$, so $\sum_{k=1}^d x^\star_k \;=\; S^\star.$ For every permutation matrix $P$, the symmetry assumption implies $P x^\star \in \Lambda$.
Let $\mathcal{U}$ be the (finite) set of all permutation matrices of size $d \times d$ and consider the averaged vector $\bar{x} \;\triangleq\; \frac{1}{|\mathcal{U}|} \sum_{P \in \mathcal{U}} P x^\star.$ By convexity of $\Lambda$, we have $\bar{x} \in \Lambda$. Moreover, each component of $\bar{x}$ is equal: for any indices $k,\ell$, there exists a permutation $P \in \mathcal{U}$ that maps coordinate k to $\ell$. Hence all coordinates of $\bar{x}$ have the same value, i.e., there exists $\kappa \in \mathbb{R}$ such that $\bar{x} = \kappa \mathbf{1}$. Moreover, permutation matrices preserve the sum of components, so $\sum_{k=1}^d (P x^\star)_k \;=\; \sum_{k=1}^d x^\star_k \;=\; S^\star,
\forall P \in \mathcal{U}.$ Taking the average over $P$ gives
\[
\sum_{k=1}^d \bar{x}_k
\;=\;
\frac{1}{|\mathcal{U}|} \sum_{P \in \mathcal{U}} \sum_{k=1}^d (P x^\star)_k
\;=\;
S^\star.
\]
Since $\bar{x} = \kappa \mathbf{1}$, we obtain $d \kappa = S^\star$, i.e., $\kappa \;=\; \frac{S^\star}{d}.$ Because $\bar{x} \in \Lambda$, this shows that $\kappa \mathbf{1} \in \Lambda$ for $\kappa = S^\star/d$, so
$\kappa^\star \ge S^\star/d$.
\end{proof}

\begin{proof}[Proof of Remark~\ref{remark:sum_throughtput_symmetric_load_connection}]
Recall that
\[
\Lambda_{\mathrm{EGS}} = \frac{p_{\mathrm{e2e}}(\beta)}{\Delta t}\, \mathrm{co}(\mathcal{X})
\subset \mathbb{R}^{F}.
\]
Since we assume that the switch is homogeneous across node pairs, $\Lambda_{\mathrm{EGS}}$ is invariant under permutations of its coordinates (i.e., under relabelings of the pairs $(i,j) \in \mathcal{F}$). That means that $\Lambda_{\mathrm{EGS}}$ is a non-empty, convex, permutation-invariant polytope in dimension
$d \triangleq F$. By definition, $R^{\mathrm{sec}}_{\mathrm{EGS}} = \max_{\lambda \in \Lambda_{\mathrm{EGS}}} \sum_{(i,j)\in\mathcal{F}} \lambda_{ij},
$
so we can identify $S^\star = R^{\mathrm{sec}}_{\mathrm{EGS}}$ in Lemma~\ref{lemma:symmetric_convex_set} for $\Lambda = \Lambda_{\mathrm{EGS}}$. Moreover, for $\Lambda = \Lambda_{\mathrm{EGS}}$ the scalar $\kappa^\star \;\triangleq\; \max\{\kappa \in \mathbb{R} : \kappa \mathbf{1} \in \Lambda_{\mathrm{EGS}}\},$ is exactly the quantity defined in Lemma~\ref{lemma:symmetric_convex_set}. Hence, applying Lemma~\ref{lemma:symmetric_convex_set} gives $\kappa^\star \;=\; \frac{R^{\mathrm{sec}}_{\mathrm{EGS}}}{F},$ which proves the claim.
\end{proof}
\else
\fi

%%%%%%%%%%%%%%%%%%%%%%%%%%%%%%%%%%%%%%%%%%%%%%

\subsection{Proof of Lemma~\ref{lem:bmatching}}
\label{sec:appendix_lemma_cardinality}

\begin{proof}
Write $E^\star \triangleq \max_{x\in\mathcal{X}}\sum_{(i,j)\in \mathcal{F}}x_{ij}$ and $T\triangleq \sum_i S_i$. We prove the lemma in two steps, first we prove that the maximum cardinality of a feasible allocation is upper bounded by $E_{\mathrm{EGS}}^{\max}$, and finally that it can achieve that value.

\emph{(I) Upper bounds:}
Every feasible allocation uses one station per selected edge, hence $E^\star \le B.$ Each edge consumes one unit of capacity at two endpoints, hence
% \begin{equation}
%     \label{UB2}
%     2\sum_{i<j}x_{ij} = \sum_i \sum_{j\neq i} x_{ij} \;\le\; \sum_i S_i \;=\; T
% \Rightarrow
% E^\star \le \Big\lfloor \tfrac{T}{2} \Big\rfloor. 
% \end{equation}
\begin{equation}
    \label{UB2}
    2\sum_{(i,j)\in \mathcal{F}}x_{ij}  \;\le\; \sum_i S_i \;=\; T
\Rightarrow
E^\star \le \Big\lfloor \tfrac{T}{2} \Big\rfloor. 
\end{equation}
% Finally, at most $S_{\max}$ incidents can be assigned to the heaviest node, and the remaining $T-S_{\max}$ units of capacity reside at the other nodes. Even if \emph{all} of that remaining capacity is paired with the heaviest node, we cannot realize more than $T-S_{\max}$ edges.
Finally, at most $S_{\max}$ edges can use the node with capacity $S_{\max}$, and the remaining $T-S_{\max}$ units of capacity reside at the other nodes. Even if all of that remaining capacity is paired with this node, we cannot realize more than $T-S_{\max}$ edges. Therefore $E^\star \le T - S_{\max}$
% \begin{equation}
%     \label{UB3}
%     E^\star \le T - S_{\max}.
% \end{equation}
and hence $E^\star \le E^{\max}_{\mathrm{EGS}}.$

\emph{(II) Achievability:}
We wish to construct an integer vector $x\in\mathcal{X}$ attaining $E^{\max}_{\mathrm{EGS}}$. For now let us ignore $B$ as a first step to construct the vector until we reintroduce it later in the proof of achievability.
Consider a graph, in which for every node $i$ we create $S_i$ vertices. A feasible allocation corresponds to pairing vertices that correspond to distinct nodes, into unordered pairs; each pair $(i,j)$ contributes one to $x_{ij}$ (multiple pairs between the same $(i,j)$ are allowed) and consumes one unit of residual capacity at both $i$ and $j$. 

We show that, without the station budget constraint $\sum_{(i,j)\in \mathcal{F}}x_{ij}\le B$, the maximum number of such pairs equals
\begin{equation}
    \label{eq:e_tilda}
\widetilde{E}\;=\;\min\!\left\{ \Big\lfloor \tfrac{T}{2} \Big\rfloor,\; T-S_{\max}\right\}.
\end{equation}
\emph{Case 1: $T-S_{\max} \le \lfloor T/2\rfloor$.}
Pick a node $i^\star$ attaining $S_{i^\star}=S_{\max}$. Pair every unit of capacity on the other nodes with $i^\star$. 
% (arbitrarily distributing repeated edges across $(i^\star,j)$, $j\neq i^\star$). 
This uses $T-S_{\max}$ pairs and respects per-node capacities: $i^\star$ contributes $T-S_{\max}\le S_{\max}$ pairs, each other node $j$ contributes at most $S_j$, and no self-pairs occur. Thus $\widetilde{E}\ge T-S_{\max}$, and by the upper bounds, equality holds.

\emph{Case 2: $T-S_{\max} > \lfloor T/2\rfloor$.}
Equivalently, $S_{\max}<\lceil T/2\rceil$. We prove that $\lfloor T/2\rfloor$ pairs are always achievable via the following greedy rule: \emph{Every node $i$ starts with residual (i.e., unassigned) capacity $S_i$. At each step, select two distinct nodes $u\neq v$ with \emph{maximum} residual capacities (ties arbitrary), add one to $x_{uv}$, and decrement both residual capacities by $1$. Stop when fewer than two nodes have positive residual capacity.}

% \begin{quote}
% At each step, select two distinct nodes $u\neq v$ with \emph{maximum} residual capacities (ties arbitrary), add one to $x_{uv}$, and decrement both residual capacities by $1$. Stop when fewer than two nodes have positive residual.
% \end{quote}

Let the residual capacities after $k$ steps be $r^{(k)}_1\ge r^{(k)}_2\ge \dots \ge r^{(k)}_n\ge 0$, and let $T^{(k)}=\sum_i r^{(k)}_i=T-2k$.
Define the imbalance $\Delta^{(k)} \;\triangleq\; r^{(k)}_1 \;-\; \sum_{i\ge 2} r^{(k)}_i.$
% \[
% \Delta^{(k)} \;\triangleq\; r^{(k)}_1 \;-\; \sum_{i\ge 2} r^{(k)}_i.
% \]
Since $S_{\max}<\lceil T/2\rceil$, $\Delta^{(0)} \;=\; S_{\max}-(T-S_{\max}) < T -2\lfloor T/2 \rfloor \;<\; 0.$
% \[
% \Delta^{(0)} \;=\; S_{\max}-(T-S_{\max}) < T -2\lfloor T/2 \rfloor \;<\; 0.
% \]
Each greedy step decreases $r^{(k)}_1$ and $\sum_{i\ge 2} r^{(k)}_i$ by exactly $1$, hence $\Delta^{(k+1)}=\Delta^{(k)}$ for all $k$ while the process runs.

Therefore $\Delta^{(k)}<0$ throughout. If at some stage $T^{(k)}\ge 2$ but $r^{(k)}_2=0$, then only one node is positive, so $\sum_{i\ge 2} r^{(k)}_i=0$ and $r^{(k)}_1>0$, implying $\Delta^{(k)}>0$, a contradiction. Thus whenever $T^{(k)}\ge 2$ we must have $r^{(k)}_2\ge 1$, and the greedy step is always feasible. Consequently, the algorithm performs exactly $\lfloor T/2\rfloor$ steps (until $T^{(k)}\in\{0,1\}$). Each step produces one valid edge between distinct nodes, so we obtain $\lfloor T/2\rfloor$ pairs. Therefore $\widetilde{E}\ge \lfloor T/2\rfloor$, and together with upper bound \eqref{UB2} we have $\widetilde{E}=\lfloor T/2\rfloor$. Therefore \eqref{eq:e_tilda} holds.

What is left now to prove the achievability step of the proof is to reintroduce $B$, since in \eqref{eq:e_tilda} we assumed no BSM budget constraints. Let $E_0 \triangleq \widetilde{E}$ be the number of pairs constructed without BSM constraints. If $B\ge E_0$, we are done. If $B<E_0$, simply truncate the construction after $B$ pairs. Truncation preserves feasibility. Hence, we have explicitly built $x\in\mathcal{X}$ with $\sum_{(i,j)\in \mathcal{F}}x_{ij}= \min\!\left\{B,\; \widetilde{E}\right\}
= E^{\max}_{\mathrm{EGS}}.$
% \[
% \sum_{i<j}x_{ij}= \min\!\left\{B,\; \widetilde{E}\right\}
% = E^{\max}_{\mathrm{EGS}}.
% \]
Together with the upper bounds this proves $E^\star=E^{\max}_{\mathrm{EGS}}$.
\end{proof}

\subsection{Proof of Theorem~\ref{thm:egs-sum-throughput}}
\label{sec:appendix_theorem_total_throughput}
\begin{proof}
The objective $\sum_{(i,j)\in \mathcal{F}}\lambda_{ij}$ is linear over the polytope $\mathrm{co}(\mathcal{X})$, so the maximum is attained at a vertex, which is an integral allocation $x^\star\in\mathcal{X}$. Therefore,
\[ \vspace{-4pt}
R^{\mathrm{sec}}_{\mathrm{EGS}}
\;=\;
\frac{p_{\mathrm{e2e}}(\beta)}{\Delta t}\,\max_{x\,\in\,\mathcal{X}} \sum_{(i,j)\in \mathcal{F}} x_{ij}
\;=\;
\frac{p_{\mathrm{e2e}}(\beta)}{\Delta t}\,E^{\max}_{\mathrm{EGS}},
\] \vspace{-3pt}
where the last equality follows from Lemma~\ref{lem:bmatching}. 
\end{proof}

\bibliographystyle{IEEEtran}

\bibliography{my_bib.bib}

% Generated by IEEEtran.bst, version: 1.14 (2015/08/26)
\begin{thebibliography}{10}
\providecommand{\url}[1]{#1}
\csname url@samestyle\endcsname
\providecommand{\newblock}{\relax}
\providecommand{\bibinfo}[2]{#2}
\providecommand{\BIBentrySTDinterwordspacing}{\spaceskip=0pt\relax}
\providecommand{\BIBentryALTinterwordstretchfactor}{4}
\providecommand{\BIBentryALTinterwordspacing}{\spaceskip=\fontdimen2\font plus
\BIBentryALTinterwordstretchfactor\fontdimen3\font minus \fontdimen4\font\relax}
\providecommand{\BIBforeignlanguage}[2]{{%
\expandafter\ifx\csname l@#1\endcsname\relax
\typeout{** WARNING: IEEEtran.bst: No hyphenation pattern has been}%
\typeout{** loaded for the language `#1'. Using the pattern for}%
\typeout{** the default language instead.}%
\else
\language=\csname l@#1\endcsname
\fi
#2}}
\providecommand{\BIBdecl}{\relax}
\BIBdecl

\bibitem{van2014quantum}
R.~Van~Meter, \emph{Quantum networking}.\hskip 1em plus 0.5em minus 0.4em\relax John Wiley \& Sons, 2014.

\bibitem{caleffi2024distributed}
M.~Caleffi \emph{et~al.}, ``Distributed quantum computing: a survey,'' \emph{Computer Networks}, vol. 254, p. 110672, 2024.

\bibitem{briegel1998quantum}
H.-J. Briegel, W.~D{\"u}r, J.~I. Cirac, and P.~Zoller, ``Quantum repeaters: the role of imperfect local operations in quantum communication,'' \emph{Physical Review Letters}, vol.~81, no.~26, p. 5932, 1998.

\bibitem{vardoyan2021stochastic}
G.~Vardoyan \emph{et~al.}, ``On the stochastic analysis of a quantum entanglement distribution switch,'' \emph{IEEE Transactions on Quantum Engineering}, vol.~2, pp. 1--16, 2021.

\bibitem{bhambay2025optimal}
S.~Bhambay, T.~Vasantam, and N.~Walton, ``Optimal scheduling in a quantum switch,'' \emph{arXiv preprint arXiv:2501.05380}, 2025.

\bibitem{einstein1935can}
A.~Einstein, B.~Podolsky, and N.~Rosen, ``Can quantum-mechanical description of physical reality be considered complete?'' \emph{Physical review}, vol.~47, no.~10, p. 777, 1935.

\bibitem{gauthier2023architecture}
S.~Gauthier, G.~Vardoyan, and S.~Wehner, ``An architecture for control of entanglement generation switches in quantum networks,'' \emph{IEEE Transactions on Quantum Engineering}, vol.~4, pp. 1--17, 2023.

\bibitem{azuma2015all}
K.~Azuma, K.~Tamaki, and H.-K. Lo, ``All-photonic quantum repeaters,'' \emph{Nature communications}, vol.~6, no.~1, p. 6787, 2015.

\bibitem{gauthier2024demand}
S.~Gauthier, T.~Vasantam, and G.~Vardoyan, ``An on-demand resource allocation algorithm for a quantum network hub and its performance analysis,'' in \emph{2024 IEEE QCE}, vol.~1.\hskip 1em plus 0.5em minus 0.4em\relax IEEE, 2024, pp. 1748--1759.

\bibitem{yau2025service}
G.~X. Yau, T.~Vasantam, and G.~Vardoyan, ``Service-the-longest-queue among d choices policy for quantum entanglement switching,'' in \emph{2025 International Conference on Quantum Communications, Networking, and Computing (QCNC)}.\hskip 1em plus 0.5em minus 0.4em\relax IEEE, 2025, pp. 1--8.

\bibitem{vasantam2022throughput}
T.~Vasantam and D.~Towsley, ``A throughput optimal scheduling policy for a quantum switch,'' in \emph{Quantum Computing, Communication, and Simulation II}, vol. 12015.\hskip 1em plus 0.5em minus 0.4em\relax SPIE, 2022, pp. 14--23.

\bibitem{tillman2024calculating}
I.~Tillman \emph{et~al.}, ``Calculating the capacity region of a quantum switch,'' in \emph{2024 IEEE QCE}, vol.~1.\hskip 1em plus 0.5em minus 0.4em\relax IEEE, 2024, pp. 1868--1878.

\bibitem{dai2021entanglement}
W.~Dai, A.~Rinaldi, and D.~Towsley, ``Entanglement swapping in quantum switches: Protocol design and stability analysis,'' \emph{arXiv preprint arXiv:2110.04116}, 2021.

\bibitem{promponas2024maximizing}
P.~Promponas \emph{et~al.}, ``Maximizing entanglement rates via efficient memory management in flexible quantum switches,'' \emph{IEEE Journal on Selected Areas in Communications}, vol.~42, no.~7, pp. 1749--1762, 2024.

\bibitem{valls2023capacity}
V.~Valls, P.~Promponas, and L.~Tassiulas, ``On the capacity of the quantum switch with and without entanglement decoherence,'' \emph{IEEE Communications Letters}, vol.~27, no.~9, pp. 2388--2392, 2023.

\bibitem{panigrahy2023capacity}
N.~K. Panigrahy \emph{et~al.}, ``On the capacity region of a quantum switch with entanglement purification,'' in \emph{IEEE INFOCOM 2023}.\hskip 1em plus 0.5em minus 0.4em\relax IEEE, 2023, pp. 1--10.

\bibitem{vardoyan2023quantum}
G.~Vardoyan and S.~Wehner, ``Quantum network utility maximization,'' in \emph{2023 IEEE QCE}, vol.~1.\hskip 1em plus 0.5em minus 0.4em\relax IEEE, 2023, pp. 1238--1248.

\bibitem{alshowkan2021reconfigurable}
M.~Alshowkan \emph{et~al.}, ``Reconfigurable quantum local area network over deployed fiber,'' \emph{PRX Quantum}, vol.~2, no.~4, Oct. 2021.

\bibitem{barrett2005efficient}
S.~D. Barrett and P.~Kok, ``Efficient high-fidelity quantum computation using matter qubits and linear optics,'' \emph{Physical Review A}, vol.~71, no.~6.

\bibitem{chen2023zero}
K.~C. Chen \emph{et~al.}, ``Zero-added-loss entangled-photon multiplexing for ground- and space-based quantum networks,'' \emph{Physical Review Applied}, vol.~19, no.~5, May 2023.

\bibitem{kok2000postselected}
P.~Kok and S.~L. Braunstein, ``Postselected versus nonpostselected quantum teleportation using parametric down-conversion,'' \emph{Physical Review A}, vol.~61, no.~4, Mar. 2000.

\bibitem{jones2016design}
C.~Jones \emph{et~al.}, ``Design and analysis of communication protocols for quantum repeater networks,'' \emph{New Journal of Physics}, vol.~18, no.~8, p. 083015, 2016.

\bibitem{gauthier2023control}
S.~Gauthier, G.~Vardoyan, and S.~Wehner, ``A control architecture for entanglement generation switches in quantum networks,'' in \emph{Proceedings of the 1st Workshop on Quantum Networks and Distributed Quantum Computing}, 2023, pp. 38--44.

\bibitem{schrijver2003combinatorial}
A.~Schrijver \emph{et~al.}, \emph{Combinatorial optimization: polyhedra and efficiency}.\hskip 1em plus 0.5em minus 0.4em\relax Springer, vol.~24, no.~2.

\bibitem{promponas2023full}
P.~Promponas, V.~Valls, and L.~Tassiulas, ``Full exploitation of limited memory in quantum entanglement switching,'' in \emph{2023 IFIP Networking}.\hskip 1em plus 0.5em minus 0.4em\relax IEEE, 2023, pp. 1--9.

\bibitem{dhara2022heralded}
P.~Dhara \emph{et~al.}, ``Heralded multiplexed high-efficiency cascaded source of dual-rail entangled photon pairs using spontaneous parametric down-conversion,'' \emph{Physical Review Applied}, vol.~17, no.~3, p. 034071, 2022.

\end{thebibliography}

\end{document}